\UseRawInputEncoding
\documentclass[11pt,reqno]{amsart}
\allowdisplaybreaks[4]
\usepackage{graphicx} 
\usepackage{subfigure}
\usepackage{epsfig}
\usepackage{amssymb}
\usepackage{amsmath,xypic}
\usepackage{cite,color}
\usepackage{cite}
\usepackage{ytableau}

\numberwithin{equation}{section}
\newtheorem{theorem}{Theorem}[section]
\newtheorem{definition}[theorem]{Definition}
\newtheorem{proposition}[theorem]{Proposition}
\newtheorem{corollary}[theorem]{Corollary}

\newtheorem{lemma}[theorem]{Lemma}

\newcommand{\be}{\begin{equation}}
\newcommand{\ee}{\end{equation}}
\newcommand{\bea}{\begin{eqnarray}}
\newcommand{\eea}{\end{eqnarray}}
\newcommand{\ba}{\begin{array}}
	\newcommand{\ea}{\end{array}}
\newcommand{\bean}{\begin{eqnarray*}}
	\newcommand{\eean}{\end{eqnarray*}}

\newcommand{\superY}[1]{\begin{array}{c} \scalebox{1}{\begin{ytableau}#1 \end{ytableau} } \end{array}}

\def\lprod{\mathop{\prod{\mkern-29.5mu}{\mathbf\longleftarrow}}}
\def\rprod{\mathop{\prod{\mkern-28.0mu}{\mathbf\longrightarrow}}}

\begin{document}
\title{Tau functions and correlation functions of the bosonic universal character hierarchy }

\author{Denghui Li, Jinzhou Liu and  Zhaowen Yan$^*$}

\dedicatory {School of Mathematical Sciences, Inner Mongolia University, \\ Inner Mongolia Key Laboratory of Mathematical Modeling and Scientific Computing,  \\ Hohhot, Inner Mongolia 010021,  P. R. China }

\thanks{*Corresponding author. Email: yanzw@imu.edu.cn (Z. W. Yan)}

\begin{abstract}
This paper is concerned with the construction of the bosonic universal character (UC) hierarchy, whose tau functions are investigated within the framework of charged free bosons. The Lie algebra corresponding to the bosonic UC hierarchy is  $\mathfrak{\widehat{gl}}_{2\infty}$. It is shown that the tau functions of bosonic UC hierarchy can be represented based on a series of ordered exponential operators. Furthermore, we derive correlation functions of the bosonic UC hierarchy, which can be expressed as the product of UCs. It is worth noting that the correlation functions of the bosonic UC hierarchy are the inverse of the correlation functions of the generalized phase model.\\
\textbf{Keywords}: tau functions; bosonic UC hierarchy; correlation functions; charged free bosons; universal characters; generalized phase model\\
\textbf{ Mathematics Subject Classifications (2000)}: 17B80, 35Q55, 37K10
\end{abstract}
\maketitle
\tableofcontents

\section{Introduction}
Classical and quantum integrable systems \cite{quantumphysics,integrable_quantum,introduction_integrablesystem}, as core research areas in mathematical physics, play a vital role in nonlinear wave theory \cite{integrable_wave}, quantum field theory \cite{integrable_QFT}, and string theory \cite{integrable_string}. The work of the Kyoto School systematically reveals the profound connections between the highest weight representations of infinite-dimensional Lie algebras and  classical integrable systems \cite{Jimbo82,Jimbo2,Kac2013}, including the Kadomtsev-Petviashvili (KP), KP of type B (BKP) and modified KP (MKP) hierarchies. By means of vertex operators, Tsuda \cite{UC} constructed an infinite-dimensional integrable system denoted as the UC hierarchy, which is a generalization of the KP hierarchy. The Lie algebra corresponding to the UC hierarchy is $\mathfrak{gl}(\infty)\oplus \mathfrak{gl}(\infty)$. The phase model is a quantum integrable model with continuous $\mathbf{U(1)}$ phase symmetry, whose scalar product is the tau function of the KP hierarchy \cite{Wheeler}. Moreover, Wang and Li \cite{GPM2019} established a realization of the two-site generalized phase model for strongly correlated bosons and derived the A-model topological string partition function on $\mathbb{C}^{3}$.

Charged free bosons and charged free fermions provide the two most fundamental frameworks for both integrable systems and algebraic structures \cite{Jimbo2000,boson_integrable}. Through Friedan-Martinec-Shenker bosonization, Wang \cite{Walgrbra_boson} proved that the $\mathcal{W}_{1+\infty}$ algebra with central charge $c=-1$ is isomorphic to the tensor product of the simple $\mathcal{W}_{3}$ algebra with central charge $c=-2$ and a Heisenberg vertex algebra generated by bosons. A new integrable sysytem named bosonic KP hierarchy was proposed by charged free bosons, and the corresponding tau functions of this hierarchy were discussed \cite{boson_hierarchy}. Meanwhile, in view of  the boson-boson correspondence, a novel proof of the Borchardt identity was presented \cite{tau_function_boson}. Guo et al. \cite{mBKP_fermoinic} obtained explicit Darboux transformation operators for the modified BKP hierarchy by  fermionic approach. The tau functions and soliton solutions of the UC hierarchy have been developed with  respect to charged free fermions \cite{UC_fermion}. Cui et al. \cite{PM_fermion} provided a representation of the phase model from the perspective of fermions. The fermion representation of the generalized phase model also has been introduced, in which the scalar products of  the generalized phase model were identified as tau functions of the UC hierarchy \cite{GPM_fermion}.

The theory of symmetric functions is not only a core part of classical algebras but also a key tool in the investigation of integrable systems \cite{Mac,Macdonald_integrable,Toda_KP,symmetric_integrable}. As a basic class of symmetric functions, Schur functions serve as tau functions of the KP hierarchy \cite{Jimbo2000}. UCs \cite{universal_characters} are a generalization of Schur functions, which are precisely the tau functions of the UC hierarchy. The correlation functions of the phase model have been expressed as determinants \cite{correctionphase}, whereas its wave functions admit a representation in terms of Schur functions \cite{model-symmetricfunction}. Based on charged free fermions,  Li et al. \cite{correclation_GPM} showed that  the correlation functions of the generalized phase model can be written as the product of UCs. Taking charged free bosons and Schur functions into consideration, Jing et al. \cite{bosoncorreclation} recently derived the correlation functions of the bosonic KP  hierarchy. Inspired by the connection between the charged free bosons and symmetric functions, we will construct a new integrable system called the bosonic UC hierarchy and explore its tau functions and correlation functions.

The paper is built up as follows. Some fundamental information about charged free bosons, UCs and the Baker-Campbell-Hausdorff (BCH) formula are reviewed in Section 2. In section 3, we construct the bosonic UC hierarchy by charged free bosons and investigate corresponding tau functions. Moreover, the representation of the Lie algebra $\mathfrak{\widehat{gl}}_{2\infty}$ is also developed in this section.  Section 4 is devoted to considering the correlation functions of the bosonic UC hierarchy. The last Section offers a summary and discussion.

\section{Preliminaries}
Our discussion begins with an overview of essential concepts, including charged free bosons, UCs and the Baker-Campbell-Hausdorff (BCH) formula\cite{bosoncorreclation,tau_function_boson,universal_characters,UC,MKP_Boson,
BCH_formula}.
\subsection{Charged free bosons}
The charged free bosons are introduced by the bosonic generating functions as the formal sums
\begin{eqnarray}\label{boson_generating_function}
&&a(z)=\sum\limits_{n\in\mathbb{Z}}a_{n}z^{-n},\quad a^{\ast}(z)=\sum\limits_{n\in\mathbb{Z}}a^{\ast}_{n}z^{-n-1},\notag\\
&&a^{\prime}(z)=\sum\limits_{n\in\mathbb{Z}}a^{\prime}_{n}z^{-n},\quad a^{\prime\ast}(z)=\sum\limits_{n\in\mathbb{Z}}a^{\prime\ast}_{n}z^{-n-1},
\end{eqnarray}
here $a_{n},a^{\ast}_{n},a^{\prime}_{n}$ and $a^{\prime\ast}_{n}$ are charged free bosons satisfying the  commutation relations
\begin{eqnarray}\label{boson_commutation}
&&[a_{m},a_{n}]=[a^{\ast}_{m},a^{\ast}_{n}]=0,\ [a_{m},a^{\ast}_{n}]=\delta_{m+n,0},\notag\\
&&[a^{\prime}_{m},a^{\prime}_{n}]=[a^{\prime\ast}_{m},a^{\prime\ast}_{n}]=0,\ [a^{\prime}_{m},a^{\prime\ast}_{n}]=\delta_{m+n,0},\notag\\
&&[a_{m},a^{\prime}_{n}]=[a_{m},a^{\prime\ast}_{n}]=[a^{\prime}_{m},a_{n}]
=[a^{\prime}_{m},a^{\ast}_{n}]=0,
\end{eqnarray}
where $[A,B]=AB-BA$ and $\delta_{m+n,0} = \begin{cases}
1, &m+n=0, \\
0, & \text{otherwise}.
\end{cases}$

The vacuum vector $|\text{vac}\rangle$ and  dual vacuum vector$\langle \text{vac}|$   are characterized by the properties below
\begin{eqnarray}\label{boson annihilation relations}
&&a_{m}|\text{vac}\rangle=a^{\prime}_{m}|\text{vac}\rangle
=a^{\ast}_{m+1}|\text{vac}\rangle=a^{\prime\ast}_{m+1}|\text{vac}\rangle=0,\notag\\
&&\langle\text{vac}|a_{n}=\langle\text{vac}|a^{\prime}_{n}
=\langle\text{vac}|a^{\ast}_{n+1}=\langle\text{vac}|a^{\prime\ast}_{n+1}=0,
\end{eqnarray}
for all integers $m\geq0, n<0$.

Let $\mathcal{A}$ be the algebra generated by charged free bosons $\{a_{m},a^{\ast}_{n},a^{\prime}_{m},a^{\prime\ast}_{n}|m,n\in\mathbb{Z}\}$. By the commutation relations of bosons, we can express any element $a\in\mathcal{A}$ as a linear combination
\begin{eqnarray}
a^{\ast\alpha_{1}}_{m_{1}}\cdots a^{\ast\alpha_{r}}_{m_{r}}a^{\beta_{1}}_{n_{1}}\cdots a^{\beta_{s}}_{n_{s}}a^{\prime\ast\gamma_{1}}_{p_{1}}\cdots a^{\prime\ast\gamma_{i}}_{p_{i}}a^{\prime\eta_{1}}_{q_{1}}\cdots a^{\prime\eta_{j}}_{q_{j}}
\end{eqnarray}
for $m_{1}<\cdots<m_{r}$, $n_{1}<\cdots<n_{s}$, $p_{1}<\cdots<p_{i}$, $q_{1}<\cdots<q_{j}$ and $\alpha_{k},\beta_{l},\gamma_{b},\eta_{c}$ are positive integers.

Introduce the bosonic Fock space $\mathcal{M}$ and dual bosonic Fock space $\mathcal{M}^{\ast}$ as follows
\begin{eqnarray}
&&\mathcal{M}=\mathcal{A}\cdot|\text{vac}\rangle=\{a|\text{vac}
\rangle\ |\ a\in\mathcal{A}\},\notag\\
&&\mathcal{M}^{\ast}=\langle\text{vac}|\cdot\mathcal{A}=\{\langle\text{vac}
|a\ |\ a\in\mathcal{A}\}.
\end{eqnarray}
In terms of the annihilation relations of charged free bosons (\ref{boson annihilation relations}), the bases of bosonic Fock space $\mathcal{M}$ and dual bosonic Fock space $\mathcal{M}^{\ast}$ are
\begin{align}
\left\{a^{\prime n^{\prime}_{p}}_{i_{p}^{\prime}}
\cdots a^{\prime n^{\prime}_{1}}_{i_{1}^{\prime}}a^{\prime\ast m^{\prime}_{q}}_{j_{q}^{\prime}}
\cdots a^{\prime\ast m^{\prime}_{1}}_{j_{1}^{\prime}}a^{n_{l}}_{i_{l}}\cdots a^{n_{1}}_{i_{1}}
a^{\ast m_{k}}_{j_{k}}\cdots a^{\ast m_{1}}_{j_{1}}|\text{vac}\rangle\,\middle|\, 
    \begin{aligned}
        & i_{l} < \cdots < i_{1} < 0,\ j_{k} < \cdots < j_{1} \leq 0, \\
        & i^{\prime}_{p} < \cdots < i^{\prime}_{1} < 0,\ j^{\prime}_{q} < \cdots < j^{\prime}_{1} \leq 0.
    \end{aligned}\right\},
\end{align}
\begin{align}
\left\{\langle\text{vac}|a^{n_{l}}_{i_{l}}\cdots a^{n_{1}}_{i_{1}}
a^{\ast m_{k}}_{j_{k}}\cdots a^{\ast m_{1}}_{j_{1}}a^{\prime n^{\prime}_{p}}_{i_{p}^{\prime}}
\cdots a^{\prime n^{\prime}_{1}}_{i_{1}^{\prime}} a^{\prime\ast m^{\prime}_{q}}_{j_{q}^{\prime}}
\cdots a^{\prime\ast m^{\prime}_{1}}_{j_{1}^{\prime}}\,\middle|\, 
    \begin{aligned}
        & i_{l} > \cdots > i_{1} \geq 0,\ j_{k} > \cdots > j_{1} > 0, \\
        & i^{\prime}_{p} > \cdots > i^{\prime}_{1} \geq 0,\ j^{\prime}_{q} > \cdots > j^{\prime}_{1} > 0.
    \end{aligned}\right\},
\end{align}
here $n_i(i=1,\ldots,l)$, $m_{j}(j=1,\ldots,k)$, $n^{\prime}_{i^{\prime}}(i^{\prime}=1,\ldots,p)$ and $m^{\prime}_{j^{\prime}}(j^{\prime}=1,\ldots,q)$ are nonnegative integers.

We define the charge of the bosons and the vacuum vectors as
\begin{eqnarray}
\centering
\begin{tabular}{c c c c c}
\hline
boson & $a_{n}$ & $a_{n}^{*}$ & $a^{\prime}_{n}$ & $a_{n}^{\prime\ast}$ \\
\hline
charge  & $(1,0)$  & $(-1,0)$   & $(0,1)$  & $(0,-1)$  \\
\hline
\end{tabular}
\end{eqnarray}
\begin{eqnarray}
\text{charge of}\ |\text{vac}\rangle(\text{or}\langle\text{vac}|)=(0,0).
\end{eqnarray}
For example, the charge of $a^{\prime2}_{-3}a^{\prime}_{-2}a^{\prime3}_{-1}a^{\prime\ast}_{-2}a_{-4}a^{2}_{-1}
a^{\ast}_{-3}a^{\ast}_{-2}|\text{vac}\rangle$ is $(1,5)$.

Let $\mathcal{M}_{(c_{1},c_{2})}$ and $\mathcal{M}^{\ast}_{(c_{1},c_{2})}$ be respectively the vector subspaces of $\mathcal{M}$ and $\mathcal{M}^{\ast}$ with definite charge $(c_{1},c_{2})$. Then $\mathcal{M}$ and $\mathcal{M}^{\ast}$ decompose as  a direct sum of vector spaces
\begin{eqnarray}
&&\mathcal{M}=\bigoplus\limits_{c_{1},c_{2}\in\mathbb{Z}}\mathcal{M}_{(c_{1},c_{2})},\notag\\
&&\mathcal{M}^{\ast}=\bigoplus\limits_{c_{1},c_{2}\in\mathbb{Z}}\mathcal{M}^{\ast}
_{(c_{1},c_{2})}.
\end{eqnarray}
Specifically, the basis of $\mathcal{M}_{(0,0)}$ is
\begin{align}
&\left\{a_{-l^{\prime}}^{\prime n^{\prime}_{l^{\prime}}}
\cdots a^{\prime n^{\prime}_{1}}_{-1}a^{\prime\ast n^{\prime}_{-k^{\prime}}}_{-k^{\prime}}
\cdots a^{\prime\ast n^{\prime}_{0}}_{0}a_{-l}^{n_{l}}\cdots a_{-1}^{n_{1}}
a^{\ast n_{-k}}_{-k}\cdots a^{\ast n_{0}}_{0}|\text{vac}\rangle\,\middle|\, 
    \begin{aligned}
        & \sum\limits_{i=0}^{k}n_{-i}=\sum\limits_{j=1}^{l}n_{j}, \ \sum\limits_{i=0}^{k^{\prime}}n^{\prime}_{-i}=\sum\limits_{j=1}^{l^{\prime}}n^{\prime}_{j}.
    \end{aligned}\right\},
\end{align}
and the basis of $\mathcal{M}^{\ast}_{(0,0)}$ is
\begin{align}
&\left\{\langle\text{vac}|a^{p_{0}}_{0}\cdots a^{p_{-k}}_{k}
a^{\ast p_{1}}_{1}\cdots a^{\ast p_{l}}_{l}a^{\prime p^{\prime}_{0}}_{0}
\cdots a^{\prime p^{\prime}_{-k^{\prime}}}_{k^{\prime}} a^{\prime\ast p^{\prime}_{1}}_{1}
\cdots a^{\prime\ast p^{\prime}_{l^{\prime}}}_{l^{\prime}}\,\middle|\, 
    \begin{aligned}
        & \sum\limits_{i=0}^{k}p_{-i}=\sum\limits_{j=1}^{l}p_{j},\ \sum\limits_{i=0}^{k^{\prime}}p^{\prime}_{-i}=\sum\limits_{j=1}^{l^{\prime}}p^{\prime}_{j}.
         \end{aligned}\right\}.
\end{align}

Consider the vacuum expectation value $\langle\ , \ \rangle$ from $\mathcal{M}^{\ast}\times\mathcal{M}$ to $\mathbb{C}$ by
\begin{eqnarray}
(\langle\text{vac}|a,b|\text{vac}\rangle)\mapsto \langle\text{vac}|a\cdot b|\text{vac}\rangle=\langle ab\rangle,\quad a,b\in\mathcal{A}.
\end{eqnarray}
In particular,
\begin{eqnarray}
&&\langle\text{vac}|1|\text{vac}\rangle=1,\notag\\
&&\langle\text{vac}|a_{m}a^{\ast}_{n}|\text{vac}\rangle=\delta_{m+n,0}\theta(n\leq0),
\end{eqnarray}
where the notation $\theta(P)$ is the Boolean characteristic function of a general property $P$. If $P$ is true, $\theta(P)=1$, otherwise $\theta(P)=0$.

The normal ordering $:a_{m}a^{\ast}_{n}:$ of charged free bosons is denoted  as
\begin{eqnarray}
:a_{m}a^{\ast}_{n}: = \begin{cases}
a^{\ast}_{n}a_{m}, &n\leq0, \\
a_{m}a^{\ast}_{n}, &n>0.
\end{cases}
\end{eqnarray}
It is apparent that
\begin{eqnarray}
:a_{m}a^{\ast}_{n}:=a_{m}a^{\ast}_{n}-\langle\text{vac}|a_{m}a^{\ast}_{n}|\text{vac}\rangle.
\end{eqnarray}

For $\langle m|=\langle\text{vac}|a^{m_{0}}_{0}\cdots a^{m_{-k}}_{k}a^{\ast m_{1}}_{1}a^{\ast m_{l}}_{l}$ and $|n\rangle=a^{n_{l}}_{-l}\cdots a^{n_{1}}_{-1}a^{\ast n_{-k}}_{-k}\cdots a^{\ast n_{0}}_{0}|\text{vac}\rangle$, it is easy to show that
\begin{eqnarray}
&&\langle m|n\rangle=(-1)^{\sum\limits^{l}_{i=1}}\prod\limits^{l}_{j=-k}\delta_{m_{j},n_{j}}m_{j}!.
\end{eqnarray}
Furthermore, assuming
\begin{eqnarray}
&&\langle m;m^{\prime}|=\langle\text{vac}|a^{m_{0}}_{0}\cdots a^{m_{-k}}_{k}a^{\ast m_{1}}_{1}a^{\ast m_{l}}_{l}a^{\prime m^{\prime}_{0}}_{0}\cdots a^{\prime m^{\prime}_{-k^{\prime}}}_{k^{\prime}}a^{\prime\ast m^{\prime}_{1}}_{1}a^{\prime\ast m^{\prime}_{l^{\prime}}}_{l^{\prime}},\notag\\
&&|n^{\prime};n\rangle=a^{\prime n^{\prime}_{l^{\prime}}}_{-l^{\prime}}\cdots a^{\prime n^{\prime}_{1}}_{-1}a^{\prime \ast n^{\prime}_{-k^{\prime}}}_{-k^{\prime}}\cdots a^{\prime \ast n^{\prime
}_{0}}_{0}a^{n_{l}}_{-l}\cdots a^{n_{1}}_{-1}a^{\ast n_{-k}}_{-k}\cdots a^{\ast n_{0}}_{0}|\text{vac}\rangle,
\end{eqnarray}
 we can obtain
\begin{eqnarray}\label{double_product}
\langle m;m^{\prime}|n^{\prime};n\rangle&=&(-1)^{\sum\limits^{l}_{i=1}+\sum\limits^{l^{\prime}}_{j=1}}
\prod\limits^{l}_{p=-k}\delta_{m_{p},n_{p}}m_{p}!\prod\limits^{l^{\prime}}_{q=-k^{\prime}}
\delta_{m^{\prime}_{q},n^{\prime}_{q}}m^{\prime}_{q}!\notag\\
&=&\langle m|n\rangle\langle m^{\prime}|n^{\prime}\rangle,
\end{eqnarray}
here $\langle m^{\prime}|=\langle\text{vac}|a^{\prime m^{\prime}_{0}}_{0}\cdots a^{\prime m^{\prime}_{-k^{\prime}}}_{k^{\prime}}a^{\prime\ast m^{\prime}_{1}}_{1}a^{\prime\ast m^{\prime}_{l^{\prime}}}_{l^{\prime}}$ and $|n^{\prime}\rangle=a^{\prime n^{\prime}_{l^{\prime}}}_{-l^{\prime}}\cdots a^{\prime n^{\prime}_{1}}_{-1}a^{\prime \ast n^{\prime}_{-k^{\prime}}}_{-k^{\prime}}\cdots a^{\prime \ast n^{\prime}_{0}}_{0}|\text{vac}\rangle$.

\subsection{Universal characters and the Baker-Campbell-Hausdorff  formula}
A partition refers to a non-increasing sequence of nonnegative integers where only finitely many entries are nonzero, which is denoted by $\lambda=(\lambda_1, \lambda_2, \ldots, \lambda_r, \ldots)$. The weight of partition $\lambda$ is   $|\lambda|=\lambda_1+\lambda_2+\ldots+\lambda_r+\ldots$. Additionally, the partition $\lambda$ can also be expressed  as $\lambda=(1^{m_{1}}2^{m_{2}}\cdots r^{m_{r}}\cdots)$, the number
\begin{eqnarray}
&&m_{i}=m_{i}(\lambda)=\text{Card}\{j:\lambda_{j}=i\}
\end{eqnarray}
is called the multiplicity of $i$ in $\lambda$.

The Young diagram of the partition $\lambda$ is formed by placing $\lambda_{i}$ boxes in the $i$-th row. Let $[N,M]$ denote the rectangular Young diagram comprising $N$ rows and $M$ columns. The  notation $\lambda\subseteq[N,M]$ means the Young diagram of the partition $\lambda$ can be  contained  in the rectangle $[N,M]$. For example, the Young diagram of the partition $\lambda=(7,6,4,3,1)$ is
\begin{eqnarray}
\scalebox{1.0}{$\superY{\,&\,&\,&\,&\,&\,&\\\,&\,&\,&\,&\,&\\\,&\,&\,&\\\,&\,&\\\,}$},
\end{eqnarray}
and this partition satisfies $\lambda\subset[5,7]$.

For a pair of partitions $\lambda=(\lambda_{1},\ldots,\lambda_{l+1})$ and $\mu=(\mu_{1},\ldots,\mu_{l})$ , we say that the partition $\lambda$ interlaces $\mu$
if and only if $\lambda_{i}\geq\mu_{i}\geq\ \lambda_{i+1}$ for all $1\leq i\leq l$, which is denoted as $\lambda\succ\mu$.

Schur function $s_{\mu}\{\mathbf{x}\}$ for variables $\{\mathbf{x}\}=\{x_{1},x_{2},\ldots,x_{n}\}$ is defined by
\begin{eqnarray}\label{schur}
s_\mu\{\mathbf{x}\}=\frac{\det\left(x_{i}^{\mu_{j}+n-j}\right)_{1 \leq i,j \leq n}}{\det\left(x_{i}^{n-j}\right)_{1 \leq i,j \leq n}}.
\end{eqnarray}

It is well-known that
\begin{eqnarray}\label{schur-skew-schur}
s_\mu\{\mathbf{x}\}= \sum_{\nu \subseteq [n-1,\infty)} s_{\mu/\nu}(x_n) \, s_\nu\{x_1, \ldots, x_{n-1}\},
\end{eqnarray}	
here $s_{\mu/\nu}(x_n)$ denotes the single variable skew Schur function, which is expressed as
\begin{align}\label{skew-schur}
s_{\mu/\nu}(x_{n}) = \begin{cases}
	x_{n}^{|\mu| - |\nu|}, & \mu \succ \nu, \\
	0, & \text{otherwise}.
\end{cases}.
\end{align}

The UC $S_{[\lambda,\mu]}\{\mathbf{x,y}\}$ is given by
\begin{eqnarray}
&&S_{[\lambda,\mu]}\{\mathbf{x,y}\}=\det
\left(\begin{array}{lll}
p_{\mu_{m-i+1}+i-j}\{\mathbf{y}\},\quad \quad  \ 1\leq i\leq m\\
p_{\lambda_{i-m}-i+j}\{\mathbf{x}\},\quad \quad m+1\leq i\leq l+m
\end{array}\right)_{1\leq i,j\leq l+m},
\end{eqnarray}
where $p_{n}\{\mathbf{x}\}$ is the elementary Schur polynomial expressed as $p_n(x) = \sum\limits_{\substack{k_1 + 2k_2 + \dots + n k_n = n \\ k_1, k_2, \dots, k_n \ge 0}} \frac{x_1^{k_1} x_2^{k_2} \cdots x_n^{k_n}}{k_1! \, k_2! \, \cdots \, k_n!}$. In terms of the Littlewood-Richardson coefficients $C^{\mu}_{\nu\tau}$ and $C^{\lambda}_{\xi\tau}$, a product of the Schur functions can be written as a sum of the UCs
\begin{eqnarray}
s_{\mu}\{\mathbf{x}\}s_{\lambda}\{\mathbf{y}\}=
\sum\limits_{\tau,\nu,\xi\in\mathcal{P}}C^{\mu}_{\nu\tau}C^{\lambda}_{\xi\tau}
S_{[\nu,\xi]}\{\mathbf{x},\mathbf{y}\}.
\end{eqnarray}

 Let $P$ and $Q$ be operators on a Hilbert space. The BCH formula  indicates that $\exp(P)\exp(Q)=\exp(R)$, here the operator $R$ is expressed as
\begin{eqnarray}
R=P+Q+\frac{1}{2}[P,Q]+\frac{1}{12}[P,[P,Q]]+\frac{1}{12}[[P,Q],Q]+\cdots.
\end{eqnarray}
For positive integers $k$, the higher-order nested commutators are recursively defined as $[P^{(k)},Q]=[P,[P^{(k-1)},Q]]$ and $[P,Q^{(k)}]=[[P,Q^{(k-1)}],Q]$, with their corresponding coefficients denoted by $\alpha_{k}$ and $\beta_{k}$, respectively. It is clear that
\begin{eqnarray}
\alpha_{k}=\beta_{k}=\frac{(-1)^{k}}{k!}\mathcal{B}_{k},
\end{eqnarray}
where $\mathcal{B}_{k}$  refers to Bernoulli numbers determined by the generating
function
\begin{eqnarray} \frac{w}{e^{w}-1}=\sum\limits_{k=0}^{\infty}\mathcal{B}_{k}\frac{w^{k}}{k!},
\quad|w|<2\pi.
\end{eqnarray}
We introduce the function $w=f(t)=-\ln(1-t)=\sum\limits_{n=1}^{\infty}\frac{1}{n}t^{n}$. For any integer $m\geq1$, the expansion  $(f(t))^{m}=\sum\limits_{s\geq m}\mathcal{K}_{s}^{m}t^{s}$.
\begin{lemma}
The coefficients $\alpha_{i}$ and $\mathcal{K}_{l}^{i}$ satisfy the following identities
\begin{eqnarray}
\sum\limits_{i=1}^{l}\alpha_{i}\mathcal{K}_{l}^{i}=\frac{1}{l+1},\quad \sum\limits_{j=1}^{l}\frac{\mathcal{K}_{l}^{j}}{j!}=1.
\end{eqnarray}
\end{lemma}

\section{The bosonic universal character hierarchy and tau functions}
The aim of this section is to construct the bosonic UC hierarchy  and discuss its tau functions.  We also propose the representation of the Lie algebra $\mathfrak{\widehat{gl}}_{2\infty}$ corresponding to the bosonic UC hierarchy in terms of charged free bosons.

\subsection{The bosonic universal character hierarchy}
\begin{definition}
For an unknown function $\tau$, the system of bilinear relations
\begin{eqnarray}\label{bosoc_uc_hierarchy}
\sum\limits_{i\in\mathbb{Z}}a^{\ast}_{i}\tau\otimes a_{-i}\tau=\sum\limits_{j\in\mathbb{Z}}a^{\prime\ast}_{j}\tau\otimes a^{\prime}_{-j}\tau=0
\end{eqnarray}
is called the bosonic UC hierarchy.

According to the Eq.(\ref{boson_generating_function}), the bosonic UC hierarchy can be rewritten equivalently into the form
\begin{eqnarray}
\text{Res}_{z}a^{\ast}(z)\otimes a(z)(\tau\otimes\tau)=\text{Res}_{z}a^{\prime\ast}(z)\otimes a^{\prime}(z)(\tau\otimes\tau)=0.
\end{eqnarray}
\end{definition}

\begin{proposition}\label{vac_solution}
If $\tau\in\mathcal{M}$ is the tau function of the bosonic UC hierarchy, then $\tau=|\text{vac}\rangle$ up to a constant.
\end{proposition}

\begin{proof}
It follows from the Eq.(\ref{boson annihilation relations}) that
\begin{eqnarray}
&&\sum\limits_{i\in\mathbb{Z}}a^{\ast}_{i}|\text{vac}\rangle\otimes a_{-i}|\text{vac}\rangle\notag\\
&=&\sum\limits_{i>0}a^{\ast}_{i}|\text{vac}\rangle\otimes a_{-i}|\text{vac}\rangle+a^{\ast}_{0}|\text{vac}\rangle\otimes a_{0}|\text{vac}\rangle+\sum\limits_{i<0}a^{\ast}_{i}|\text{vac}\rangle\otimes a_{-i}|\text{vac}\rangle\notag\\
&=&0.
\end{eqnarray}
Similarly, one can show that
\begin{eqnarray}
\sum\limits_{j\in\mathbb{Z}}a^{\prime\ast}_{j}|\text{vac}\rangle\otimes a^{\prime}_{-j}|\text{vac}\rangle=0.
\end{eqnarray}
Hence, $\tau=|\text{vac}\rangle$ is a tau function of the bosonic UC hierarchy.

Suppose that $\tau\neq|\text{vac}\rangle$ is a tau function of the bosonic UC hierarchy  and can be expressed as a sum of monomials formed from the basis of the Fock space $\mathcal{M}$. Let $N>0$ and $N^{\prime}>0$ be the largest integers for which $a_{-N}$ and $a^{\prime}_{-N^{\prime}}$ appear in the tau function, respectively. Then the tau function can be represented as
\begin{eqnarray}
\tau=\sum\limits^{m}_{k=0}\sum\limits^{s}_{l=0}a^{k}_{-N}a^{\prime l}_{-N^{\prime}}P_{k,l}(a_{-N+1},\ldots,a_{-1};\ldots,a^{\ast}_{-1},a^{\ast}_{0};
a^{\prime}_{-N^{\prime}+1},\ldots,a^{\prime}_{-1};\ldots,a^{\prime\ast}_{-1},
a^{\prime\ast}_{0})|\text{vac}\rangle,
\end{eqnarray}
where $P_{m,s}\neq0$ for all $m,s\geq1$, and that each $P_{k,l}$ with $k,l\geq0$ is a  linear combination of the basis elements of the Fock space $\mathcal{M}$. Moreover, the largest indices $n$ and $n^{\prime}$ such that $a_{-n}$ and $a^{\prime}_{-n^{\prime}}$ appear in $P_{k,l}$  satisfy $n\leq N-1$ and $n^{\prime}\leq N^{\prime}-1$, respectively. Then we obtain
\begin{eqnarray}
\sum\limits_{i\in\mathbb{Z}}a^{\ast}_{i}\tau\otimes a_{-i}\tau&=&\sum\limits_{i>N}a^{\ast}_{i}\tau\otimes a_{-i}\tau+a^{\ast}_{N}\tau\otimes a_{-N}\tau+\sum\limits_{i<N}a^{\ast}_{i}\tau\otimes a_{-i}\tau\notag\\
&=&\sum\limits^{m}_{k=1}-ka^{k-1}_{-N}\sum\limits^{s}_{l=0}a^{\prime l}_{-N^{\prime}}P_{k,l}|\text{vac}\rangle\otimes\sum\limits^{m}_{k=0}a^{k+1}_{-N}
\sum\limits^{s}_{l=0}a^{\prime l}_{-N^{\prime}}P_{k,l}|\text{vac}\rangle\notag\\
&&+\sum\limits_{i<N}a^{\ast}_{i}\sum\limits^{m}_{k=0}a^{k}_{-N}\sum\limits^{s}_{l=0}
a^{\prime l}_{-N^{\prime}}P_{k,l}|\text{vac}\rangle\otimes a_{-i}\sum\limits^{m}_{k=0}a^{k}_{-N}\sum\limits^{s}_{l=0}
a^{\prime l}_{-N^{\prime}}P_{k,l}|\text{vac}\rangle.
\end{eqnarray}
It should be noted that the second summand $\sum\limits_{i<N}a^{\ast}_{i}\tau\otimes a_{-i}\tau$ can involve at most $a^{m}_{-N}$ on the right of the tensor products. However, the item $a^{\ast}_{N}\tau\otimes a_{-N}\tau$ contains at most $-a^{m+1}_{-N}$, which cannot cancel each other out with $a^{m}_{-N}$. It can be concluded that $\sum\limits_{i\in\mathbb{Z}}a^{\ast}_{i}\tau\otimes a_{-i}\tau\neq0$. By the same argument, one has $\sum\limits_{j\in\mathbb{Z}}a^{\prime\ast}_{j}\tau\otimes a^{\prime}_{-j}\tau\neq0$. Therefore, $\tau=|\text{vac}\rangle$ up to a constant.
\end{proof}

\subsection{Tau functions}
We first construct the Lie algebra corresponding to the bosonic UC hierarchy.

Let $A=(a_{ij})_{i,j\in\mathbb{Z}}$ be the infinite complex matrices satisfying
\begin{eqnarray}\label{a_martix}
&&a_{ij}=0\quad \text{for}\quad |i-j|\gg0.
\end{eqnarray}
Then $\mathfrak{\overline{gl}}_{\infty}=\{(a_{ij})_{i,j\in\mathbb{Z}}|a_{ij}=0,|i-j|\gg0\}$ forms an infinite-dimensional Lie algebra with the Lie bracket $[A,B]=AB-BA$.

Introduce the Lie algebra $\mathfrak{\widehat{gl}}_{2\infty}=\mathfrak{\overline{gl}}_{\infty}\oplus
\mathfrak{\overline{gl}}_{\infty}\oplus\mathbb{C}c$ with a one dimensional central extension by the following bracket relation
\begin{eqnarray}\label{affine_bracket}
[A\oplus A^{\prime}\oplus\alpha c,B\oplus B^{\prime}\oplus\beta c]=(AB-BA)\oplus(A^{\prime}B^{\prime}-B^{\prime}A^{\prime})\oplus
\left(\mu(A,B)+\mu(A^{\prime},B^{\prime})\right)c,
\end{eqnarray}
where $A,B,A^{\prime},B^{\prime}\in\mathfrak{\overline{gl}}_{\infty}$, $\alpha,\beta\in\mathbb{C}$ and $\mu(A,B)$ denotes a 2-cocycle on $\mathfrak{\overline{gl}}_{\infty}$ defined by
\begin{eqnarray}
\mu(E_{ij},E_{kl})=\delta_{j,k}\delta_{i,l}\left(\theta(l\leq0)-\theta(k\leq0)\right),
\end{eqnarray}
here $E_{ij}$ means an infinite matrix with the $(i,j)$-entry 1 and zeros elsewhere.

\begin{proposition}
Consider the map $\rho:\mathfrak{\widehat{gl}}_{2\infty}\rightarrow End(\mathcal{M})$ by
\begin{eqnarray}
&&\rho(\overline{E}_{ij})=-:a_{-i}a^{\ast}_{j}:,\\
&&\rho(\overline{E^{\prime}}_{ij})=-:a^{\prime}_{-i}a^{\prime\ast}_{j}:,\\
&&\rho(c)=-1,
\end{eqnarray}
where $\overline{E}_{ij}=E_{ij}\oplus\mathbf{0}, \overline{E^{\prime}}_{ij}=\mathbf{0}
\oplus{E}_{ij}\in \mathfrak{\overline{gl}}_{\infty}\oplus
\mathfrak{\overline{gl}}_{\infty}$, and $\mathbf{0}$ is the zero matrix. Then $\rho$ is the representation of $\mathfrak{\widehat{gl}}_{2\infty}$ on $\mathcal{M}$.
\end{proposition}

\begin{proof}
By means of
\begin{eqnarray}\label{boson_four_relation}
&&\left[a_{m}a^{\ast}_{n},a_{s}a^{\ast}_{t}\right]=a_{s}a^{\ast}_{n}\delta_{m+t,0}
-a_{m}a^{\ast}_{t}\delta_{n+s,0},
\end{eqnarray}
we obtain
\begin{eqnarray}
[\rho(\overline{E}_{ij}),\rho(\overline{E}_{kj})]
&=&\left[-:a_{-i}a^{\ast}_{j}:,
-:a_{-k}a^{\ast}_{l}:\right]\notag\\
&=&\left[a_{-i}a^{\ast}_{j}-\langle\text{vac}|a_{-i}a^{\ast}_{j}|\text{vac}\rangle,
a_{-k}a^{\ast}_{l}-\langle\text{vac}|a_{-k}a^{\ast}_{k}|\text{vac}\rangle\right]\notag\\
&=&\left[a_{-i}a^{\ast}_{j},a_{-k}a^{\ast}_{l}\right]\notag\\
&=&a_{-k}a^{\ast}_{j}\delta_{i,l}-a_{-i}a^{\ast}_{l}\delta_{k,j}\notag\\
&=&\delta_{i,l}\left(:a_{-k}a^{\ast}_{j}:+\langle\text{vac}|a_{-k}a^{\ast}_{j}|
\text{vac}\rangle\right)-\delta_{j,k}\left(:a_{-i}a^{\ast}_{l}:+\langle\text{vac}|a_{-i}
a^{\ast}_{l}|\text{vac}\rangle\right)\notag\\
&=&\delta_{j,k}\rho(\overline{E}_{il})-\delta_{i,l}\rho(\overline{E}_{kj})
+\delta_{i,l}\delta_{j,k}\left(\theta(j\leq0)-\theta(l\leq0)\right).
\end{eqnarray}
It is obvious that
\begin{eqnarray}
[\overline{E}_{ij},\overline{E}_{kj}]=\delta_{j,k}\overline{E}_{il}
-\delta_{i,l}\overline{E}_{kj}+\delta_{i,l}\delta_{j,k}
\left(\theta(l\leq0)-\theta(j\leq0)\right)c.
\end{eqnarray}
As a result,
\begin{eqnarray}
[\rho(\overline{E}_{ij}),\rho(\overline{E}_{kj})]=\rho\left(
[\overline{E}_{ij},\overline{E}_{kj}]\right).
\end{eqnarray}
Similarly, the following  identity
\begin{eqnarray}
[\rho(\overline{E^{\prime}}_{ij}),\rho(\overline{E^{\prime}}_{kj})]=\rho\left(
[\overline{E^{\prime}}_{ij},\overline{E^{\prime}}_{kj}]\right)
\end{eqnarray}
holds. Furthermore, it follows from the Eqs.(\ref{boson_commutation}) and (\ref{affine_bracket}) that
\begin{eqnarray}
[\rho(\overline{E}_{ij}),\rho(\overline{E^{\prime}}_{kj})]
=\rho\left([\overline{E}_{ij},\overline{E^{\prime}}_{kj}]\right)=0.
\end{eqnarray}
Therefore, $\rho$ is the representation of $\mathfrak{\widehat{gl}}_{2\infty}$ on $\mathcal{M}$.
\end{proof}

We regard the Lie algebra $\mathfrak{\widehat{gl}}_{2\infty}$ as $\mathfrak{\widehat{gl}}_{2\infty}=\big\{X=\sum\limits_{i,j\in\mathbb{Z}}(
c_{ij}:a_{-i}a^{\ast}_{j}:+c^{\prime}_{ij}:a^{\prime}_{-i}a^{\prime\ast}_{j}:)
+c\big\}$, where $c_{ij}$ and $c^{\prime}_{ij}$ satisfy (\ref{a_martix}), and $c\in\mathbb{C}$. Consider the Lie group $\mathbf{G}$ associated with $\mathfrak{\widehat{gl}}_{2\infty}$
\begin{eqnarray}
\mathbf{G}=\left\{e^{X_{1}}e^{X_{2}}\cdots e^{X_{k}}|X_{i}\in\mathfrak{\widehat{gl}}_{2\infty}\right\}.
\end{eqnarray}
It is important to investigate the orbit of vacuum vector with respect to the action of  $\mathbf{G}$
\begin{eqnarray}
\mathbf{G}|\text{vac}\rangle=\left\{g|\text{vac}\ | g\in\mathbf{G} \right\}\subset \mathcal{M}_{(0,0)}\subset\mathcal{M}.
\end{eqnarray}

Based on the vector space $\mathcal{M}$ and dual vector space $\mathcal{M}^{\ast}$, introduce the completion vector spaces $\mathcal{\widetilde{M}}=\mathbb{C}[[a_{-i-1},a^{\ast}_{-i},a^{\prime}_{-j-1},
a^{\prime\ast}_{-j}|i,j\geq0]]|\text{vac}\rangle$ and $\mathcal{\widetilde{M}}^{\ast}=\langle\text{vac}|\mathbb{C}[[a_{i},a^{\ast}_{i+1},
a^{\prime}_{j},a^{\ast\prime}_{j+1}|i,j\geq0]]$. Next, the tau function will be discussed in  the completion vector space.

\begin{proposition}
The function
\begin{eqnarray}\label{tau_com}
\tau=\exp\big(\sum\limits_{i\geq0\atop{j>0}}c_{ij}a^{\ast}_{-i}a_{-j}+
\sum\limits_{k\geq0\atop{l>0}}c^{\prime}_{kl}a^{\prime\ast}_{-k}a^{\prime}_{-l}\big)
|\text{vac}\rangle \in \mathcal{\widetilde{M}}
\end{eqnarray}
is a tau function of the the bosonic UC hierarchy, where $c_{ij}\neq0, c^{\prime}_{kl}\neq0$.
\end{proposition}

\begin{proof}
In terms of the commutation relations of charged free bosons (\ref{boson_commutation}), it is apparent that
\begin{eqnarray}\label{solution_jiaohuan}
&&\big[\sum\limits_{i\in\mathbb{Z}}a^{\ast}_{i}\otimes a_{-i},1\otimes(a^{\ast}_{m}a_{n}+a^{\prime\ast}_{p}a^{\prime}_{q})
+(a^{\ast}_{m}a_{n}+a^{\prime\ast}_{p}a^{\prime}_{q})\otimes1\big]\notag\\
&=&\sum\limits_{i\in\mathbb{Z}}\left(\left[a^{\ast}_{i}\otimes a_{-i},1\otimes(a^{\ast}_{m}a_{n}+a^{\prime\ast}_{p}a^{\prime}_{q})\right]
+\left[a^{\ast}_{i}\otimes a_{-i},(a^{\ast}_{m}a_{n}+a^{\prime\ast}_{p}a^{\prime}_{q})\otimes1\right]\right)\notag\\
&=&\sum\limits_{i\in\mathbb{Z}}\left(a^{\ast}_{i}\otimes \left[a_{-i},a^{\ast}_{m}a_{n}+a^{\prime\ast}_{p}a^{\prime}_{q}\right]+
\left[a^{\ast}_{i},a^{\ast}_{m}a_{n}+a^{\prime\ast}_{p}a^{\prime}_{q}\right]\otimes a_{-i}\right)\notag\\
&=&\sum\limits_{i\in\mathbb{Z}}a^{\ast}_{i}\otimes\left([a_{-i},a^{\ast}_{m}]a_{n}+
a^{\ast}_{m}[a_{-i},a_{n}]+[a_{-i},a^{\prime\ast}_{p}]a^{\prime}_{q}+
a^{\prime\ast}_{p}[a_{-i},a^{\prime}_{q}]\right)\notag\\
&&+\sum\limits_{i\in\mathbb{Z}}\left([a^{\ast}_{i},a^{\ast}_{m}]a_{n}
+a^{\ast}_{m}[a^{\ast}_{i},a_{n}]+a^{\prime\ast}_{p}[a^{\ast}_{i},a^{\prime}_{q}]
+[a^{\ast}_{i},a^{\prime\ast}_{p}]a^{\prime}_{q}\right)\otimes a_{-i}\notag\\
&=&\sum\limits_{i\in\mathbb{Z}}a^{\ast}_{i}\otimes \delta_{-i+m,0}a_{n}-\sum\limits_{i\in\mathbb{Z}}a^{\ast}_{m}\delta_{i+n,0}\otimes a_{-i}\notag\\
&=&0.
\end{eqnarray}
Similarly, $\big[\sum\limits_{j\in\mathbb{Z}}a^{\prime\ast}_{j}\otimes a^{\prime}_{-j},1\otimes(a^{\ast}_{m}a_{n}+a^{\prime\ast}_{p}a^{\prime}_{q})
+(a^{\ast}_{m}a_{n}+a^{\prime\ast}_{p}a^{\prime}_{q})\otimes1\big]=0$.
Furthermore, we have
\begin{align}
&\sum\limits_{i\in\mathbb{Z}}a^{\ast}_{i}\tau\otimes a_{-i}\tau\notag\\
&=
\sum\limits_{i\in\mathbb{Z}}a^{\ast}_{i}\exp\big(\sum\limits_{i\geq0\atop{j>0}}c_{ij}a^{\ast}_{-i}a_{-j}+
\sum\limits_{k\geq0\atop{l>0}}c^{\prime}_{kl}a^{\prime\ast}_{-k}a^{\prime}_{-l}\big)
|\text{vac}\rangle\otimes a_{-i}\exp\big(\sum\limits_{i\geq0\atop{j>0}}c_{ij}a^{\ast}_{-i}a_{-j}+
\sum\limits_{k\geq0\atop{l>0}}c^{\prime}_{kl}a^{\prime\ast}_{-k}a^{\prime}_{-l}\big)
|\text{vac}\rangle\notag\\
&=\exp\big(\sum\limits_{i\geq0\atop{j>0}}c_{ij}a^{\ast}_{-i}a_{-j}+
\sum\limits_{k\geq0\atop{l>0}}c^{\prime}_{kl}a^{\prime\ast}_{-k}a^{\prime}_{-l}\big)
\otimes\exp\big(\sum\limits_{i\geq0\atop{j>0}}c_{ij}a^{\ast}_{-i}a_{-j}+
\sum\limits_{k\geq0\atop{l>0}}c^{\prime}_{kl}a^{\prime\ast}_{-k}a^{\prime}_{-l}\big)
\sum\limits_{i\in\mathbb{Z}}a^{\ast}_{i}|\text{vac}\rangle\otimes a_{-i}|\text{vac}\rangle\notag\\
&=0.
\end{align}
Then, the identity (\ref{tau_com}) is a tau function of the the bosonic UC hierarchy.
\end{proof}

\section{Correlation functions of the bosonic universal character hierarchy}
In this section, we firstly construct a class of ordered exponential operators based on charged free bosons and deduce the relations among these operators. The tau functions of the bosonic UC hierarchy have been presented in view of these operators. Moreover, by means of UCs and  the vacuum expectation value, the correlation functions of the bosonic UC hierarchy have been obtained.

Introduce the ordered exponential operators $\mathbf{B}^{\ast}_{1}(x),\ \mathbf{B}^{\ast}_{2}(y),\ \mathbf{C}^{\ast}_{1}(x)$ and $\mathbf{C}^{\ast}_{2}(y)$  by charged free bosons
\begin{eqnarray}
&&\mathbf{B}^{\ast}_{1}(x)=\exp(xa_{-M_{1}}a^{\ast}_{M_{1}-1})
\exp(xa_{-M_{1}+1}a^{\ast}_{M_{1}-2})\cdots=\rprod\limits
_{i\in(-\infty,M_{1}]}\exp(xa_{-i}a^{\ast}_{i-1}),\notag\\
&&\mathbf{B}^{\ast}_{2}(y)=\exp(ya^{\prime}_{-M_{2}}a^{\prime\ast}_{M_{2}-1})
\exp(ya^{\prime}_{-M_{2}+1}a^{\prime\ast}_{M_{2}-2})\cdots=\rprod\limits
_{i\in(-\infty,M_{2}]}\exp(ya^{\prime}_{-i}a^{\prime\ast}_{i-1}),\notag\\
&&\mathbf{C}^{\ast}_{1}(x)=\cdots\exp(xa_{-M_{1}+2}a^{\ast}_{M_{1}-1})
\exp(xa_{-M_{1}+1}a^{\ast}_{M_{1}})=\lprod\limits_{i\in(-\infty,M_{1}]}
\exp(xa_{-i+1}a^{\ast}_{i}),\notag\\
&&\mathbf{C}^{\ast}_{2}(y)=\cdots\exp(ya^{\prime}_{-M_{2}+2}a^{\prime\ast}_{M_{2}-1})
\exp(ya^{\prime}_{-M_{2}+1}a^{\prime\ast}_{M_{2}})=\lprod\limits_{i\in(-\infty,M_{2}]}
\exp(ya^{\prime}_{-i+1}a^{\prime\ast}_{i}),
\end{eqnarray}
where $M_{1},M_{2}\in\mathbb{N}$, and the product $\rprod\limits_{i\in[a,b]}$ (respectively $\lprod\limits_{i\in[a,b]}$) is the ordered product evaluated from left to right (respectively right to left), achieved by decreasing $i$ from $b$ to $a$.

\begin{proposition}\label{Bzong}
For $M_{1},M_{2}\in\mathbb{N}$, it follows that
\begin{eqnarray}
&&\mathbf{B}^{\ast}_{1}(x)=\exp\left(\sum\limits^{\infty}_{n=1}(-1)^{n+1}
\frac{x^{n}}{n}\sum\limits^{M_{1}}_{i=-\infty}a_{-i}a^{\ast}_{i-n}\right),\notag\\
&&\mathbf{B}^{\ast}_{2}(y)=\exp\left(\sum\limits^{\infty}_{n=1}(-1)^{n+1}
\frac{y^{n}}{n}\sum\limits^{M_{2}}_{i=-\infty}a^{\prime}_{-i}
a^{\prime\ast}_{i-n}\right),\notag\\
&&\mathbf{C}^{\ast}_{1}(x)=\exp\left(\sum\limits^{\infty}_{n=1}(-1)^{n+1}
\frac{x^{n}}{n}\sum\limits^{M_{1}}_{i=-\infty}a_{-i+n}a^{\ast}_{i}\right),\notag\\
&&\mathbf{C}^{\ast}_{2}(y)=\exp\left(\sum\limits^{\infty}_{n=1}(-1)^{n+1}
\frac{y^{n}}{n}\sum\limits^{M_{2}}_{i=-\infty}a^{\prime}_{-i+n}a^{\prime\ast}_{i}\right).
\end{eqnarray}
\end{proposition}

\begin{proof}
For a finite interval $i\in[m,M_{1}]$, the ordered exponential operator $\mathbf{B}^{\ast}_{1}(x)$ is expressed as
\begin{eqnarray}\label{B_1_m_M}
\mathbf{B}^{\ast}_{1}(x)=\exp(xa_{-M_{1}}a^{\ast}_{M_{1}-1})
\cdots\exp(xa_{-m}a^{\ast}_{m-1})=\rprod\limits
_{i\in[m,M_{1}]}\exp(xa_{-i}a^{\ast}_{i-1}).
\end{eqnarray}
By induction on $m$, we prove
\begin{eqnarray}\label{B_induction}
\mathbf{B}^{\ast}_{1}(x)=\exp\left(\sum\limits^{M_{1}-m+1}_{n=1}(-1)^{n+1}
\frac{x^{n}}{n}\sum\limits^{M_{1}}_{j=n+m-1}a_{-j}a^{\ast}_{j-n}\right).
\end{eqnarray}
When $m=M_{1}$, the Eqs.(\ref{B_1_m_M}) and (\ref{B_induction})  are clearly equal. Suppose the Eq.(\ref{B_induction}) holds for $m=k+1$. It is convenient to denote
\begin{eqnarray}
&&P=\sum\limits^{M_{1}-k}_{n=1}(-1)^{n+1}
\frac{x^{n}}{n}\sum\limits^{M_{1}}_{j=n+k}a_{-j}a^{\ast}_{j-n},\\
&&Q=xa_{-k}a^{\ast}_{k-1}.
\end{eqnarray}

For $m=k$, the Eq.(\ref{B_1_m_M}) is given by $\mathbf{B}^{\ast}_{1}(x)=\exp(P)\exp(Q)$. According to the Eq.(\ref{boson_four_relation}), the higher-order nested commutators $[P^{(l)},Q]$ are
\begin{eqnarray}
&&[P^{(l)},Q]=\sum\limits^{M_{1}-k-l}_{j=0}(-1)^{j+l}\mathcal{K}^{l}_{l+j}x^{l+j+1}
a_{-k-l-j}a^{\ast}_{k-1},\quad 1\leq l\leq M_{1}-k.
\end{eqnarray}
It follows from the BCH formula that
\begin{eqnarray}
R&=&P+Q+\sum\limits^{M_{1}-k}_{l=1}\alpha_{l}\sum\limits^{M_{1}-k-l}_{j=0}(-1)^{j+l}
\mathcal{K}^{l}_{l+j}x^{l+j+1}a_{-k-l-j}a^{\ast}_{k-1}\notag\\
&=&P+Q+\sum\limits^{M_{1}-k}_{s=1}\sum\limits^{s}_{l=1}\alpha_{l}(-1)^{s}
\mathcal{K}^{l}_{s}x^{s+1}a_{-k-s}a^{\ast}_{k-1}\notag\\
&=&\sum\limits^{M_{1}-k}_{n=1}(-1)^{n+1}
\frac{x^{n}}{n}\sum\limits^{M_{1}}_{j=n+k}a_{-j}a^{\ast}_{j-n}+xa_{-k}a^{\ast}_{k-1}
+\sum\limits^{M_{1}-k}_{s=1}(-1)^{s}\frac{x^{s+1}}{s+1}a_{-k-s}a^{\ast}_{k-1}\notag\\
&=&\sum\limits^{M_{1}-k+1}_{n=1}(-1)^{n+1}\frac{x^{n}}{n}\sum\limits^{M_{1}}_{j=n+k-1}
a_{-j}a^{\ast}_{j-n}.
\end{eqnarray}
This confirms the inductive hypothesis holds at $m=k$. In the limit $m\rightarrow-\infty$, the finite sum $\sum\limits^{M_{1}}_{j=n+k-1}
a_{-j}a^{\ast}_{j-n}$ converges to the infinite series $\sum\limits^{M_{1}}_{j=-\infty}
a_{-j}a^{\ast}_{j-n}$. Other expressions in the proposition can be established in a similar manner, and their derivations are omitted for brevity.
\end{proof}

\begin{proposition}
For $M_{1},M_{2}\in\mathbb{N}$, it is apparent that
\begin{eqnarray}
&&[\mathbf{B}^{\ast}_{i}(x),\mathbf{B}^{\ast}_{j}(y)]
=[\mathbf{C}^{\ast}_{i}(x),\mathbf{C}^{\ast}_{j}(y)]=0,\quad i,j=1,2.
\end{eqnarray}
\end{proposition}

\begin{proof}
From
\begin{eqnarray}
&&\left[\sum\limits^{\infty}_{n=1}(-1)^{n+1}\frac{x^{n}}{n}\sum\limits^{M_{1}}
_{k=-\infty}a_{-k}a^{\ast}_{k-n},\sum\limits^{\infty}_{m=1}(-1)^{m+1}
\frac{y^{m}}{m}\sum\limits^{M_{1}}
_{l=-\infty}a_{-l}a^{\ast}_{l-m}\right]\notag\\
&=&\sum\limits^{\infty}_{n=1}(-1)^{n+1}\frac{x^{n}}{n}\sum\limits^{\infty}_{m=1}(-1)^{m+1}
\frac{y^{m}}{m}\sum\limits^{M_{1}}_{k=-\infty}\sum\limits^{M_{1}}_{l=-\infty}
\left(a_{-l}a^{\ast}_{k-n}\delta_{-k+l-m,0}-a_{-k}a^{\ast}_{l-m}\delta_{k-n-l,0}\right)\notag\\
&=&\sum\limits^{\infty}_{n=1}(-1)^{n+1}\frac{x^{n}}{n}\sum\limits^{\infty}_{m=1}(-1)^{m+1}
\frac{y^{m}}{m}\left(\sum\limits^{M_{1}}_{l=-\infty}a_{-l}a^{\ast}_{l-m-n}-\sum\limits^{M_{1}}
_{k=-\infty}a_{-k}a^{\ast}_{k-m-n}\right)\notag\\
&=&0,
\end{eqnarray}
we deduce that the identity $[\mathbf{B}^{\ast}_{1}(x),\mathbf{B}^{\ast}_{1}(y)]=0$ holds. Analogous arguments apply to the remaining expressions, whose proofs are omitted.
\end{proof}

\begin{proposition}
For $\forall N_{1}, N_{2}\in \mathbb{Z}_{+}$, the function
\begin{eqnarray}
&&\tau=\mathbf{B}^{\ast}_{2}(y_{1})\mathbf{B}^{\ast}_{2}(y_{2})\cdots
\mathbf{B}^{\ast}_{2}(y_{N_{2}})\mathbf{B}^{\ast}_{1}(x_{1})
\mathbf{B}^{\ast}_{1}(x_{2})\cdots\mathbf{B}^{\ast}_{1}(x_{N_{1}})|\text{vac}\rangle \in \mathcal{\widetilde{M}}
\end{eqnarray}
is a tau function of the the bosonic UC hierarchy.
\end{proposition}

\begin{proof}
By the Eq.(\ref{solution_jiaohuan}), it is simple to show that
\begin{eqnarray}
&&\left[\text{Res}_{z}a^{\ast}(z)\otimes a(z),\mathbf{B}^{\ast}_{2}(y)\mathbf{B}^{\ast}_{1}(x)\otimes\mathbf{B}^{\ast}_{2}(y)
\mathbf{B}^{\ast}_{1}(x)\right]=0,\notag\\
&&\left[\text{Res}_{z}a^{\prime\ast}(z)\otimes a^{\prime}(z),\mathbf{B}^{\ast}_{2}(y)\mathbf{B}^{\ast}_{1}(x)\otimes\mathbf{B}^{\ast}_{2}(y)
\mathbf{B}^{\ast}_{1}(x)\right]=0.
\end{eqnarray}
Furthermore, one can derive that
\begin{eqnarray}
&&\text{Res}_{z}a^{\ast}(z)\otimes a(z)(\tau\otimes\tau)\notag\\
&=&\text{Res}_{z}a^{\ast}(z)\otimes a(z)(\mathbf{B}^{\ast}_{2}(y_{1})\mathbf{B}^{\ast}_{2}(y_{2})\cdots
\mathbf{B}^{\ast}_{2}(y_{N_{2}})\mathbf{B}^{\ast}_{1}(x_{1})
\mathbf{B}^{\ast}_{1}(x_{2})\cdots\mathbf{B}^{\ast}_{1}(x_{N_{1}})|\text{vac}
\rangle\otimes\mathbf{B}^{\ast}_{2}(y_{1})\notag\\
&&\mathbf{B}^{\ast}_{2}(y_{2})\cdots
\mathbf{B}^{\ast}_{2}(y_{N_{2}})\mathbf{B}^{\ast}_{1}(x_{1})
\mathbf{B}^{\ast}_{1}(x_{2})\cdots\mathbf{B}^{\ast}_{1}(x_{N_{1}})|
\text{vac}\rangle)\notag\\
&=&\mathbf{B}^{\ast}_{2}(y_{1})\mathbf{B}^{\ast}_{2}(y_{2})\cdots
\mathbf{B}^{\ast}_{2}(y_{N_{2}})\mathbf{B}^{\ast}_{1}(x_{1})
\mathbf{B}^{\ast}_{1}(x_{2})\cdots\mathbf{B}^{\ast}_{1}(x_{N_{1}})\otimes
\mathbf{B}^{\ast}_{2}(y_{1})\mathbf{B}^{\ast}_{2}(y_{2})\cdots
\mathbf{B}^{\ast}_{2}(y_{N_{2}})\mathbf{B}^{\ast}_{1}(x_{1})\notag\\
&&\mathbf{B}^{\ast}_{1}(x_{2})\cdots\mathbf{B}^{\ast}_{1}(x_{N_{1}})
\left(\text{Res}_{z}a^{\ast}(z)|\text{vac}\rangle\otimes a(z)|\text{vac}\rangle\right)\notag\\
&=&0.
\end{eqnarray}
The same argument implies that the identity
\begin{eqnarray}
\text{Res}_{z}a^{\prime\ast}(z)\otimes a^{\prime}(z)(\tau\otimes\tau)=0
\end{eqnarray}
holds. Hence, $\tau=\mathbf{B}^{\ast}_{2}(y_{1})\mathbf{B}^{\ast}_{2}(y_{2})\cdots
\mathbf{B}^{\ast}_{2}(y_{N_{2}})\mathbf{B}^{\ast}_{1}(x_{1})
\mathbf{B}^{\ast}_{1}(x_{2})\cdots\mathbf{B}^{\ast}_{1}(x_{N_{1}})|\text{vac}\rangle$ is a tau function of  the bosonic UC hierarchy.
\end{proof}

\begin{theorem}
For $\forall M_{1}, M_{2},N_{1}, N_{2}\in\mathbb{Z}_{+}$, it can be inferred that
\begin{eqnarray}\label{B_1_N}
&&\mathbf{B}^{\ast}_{2}(y_{1})\mathbf{B}^{\ast}_{2}(y_{2})\cdots
\mathbf{B}^{\ast}_{2}(y_{N_{2}})\mathbf{B}^{\ast}_{1}(x_{1})
\mathbf{B}^{\ast}_{1}(x_{2})\cdots\mathbf{B}^{\ast}_{1}(x_{N_{1}})|\text{vac}\rangle\notag\\
&=&\exp\left(\sum\limits_{1\leq k\leq M_{2}\atop{1\leq n\leq N_{2}}}(-1)^{k-1}s_{(k,1^{n-1})}\{\mathbf{y}\}a^{\prime}_{-k}a^{\prime\ast}_{1-n}\right)
\exp\left(\sum\limits_{1\leq l\leq M_{1}\atop{1\leq n\leq N_{1}}}(-1)^{l-1}s_{(l,1^{n-1})}\{\mathbf{x}\}a_{-l}a^{\ast}_{1-n}\right)|\text{vac}\rangle,
\notag\\
\end{eqnarray}
where $\{\mathbf{x}\}=\{x_{1},x_{2},\ldots,x_{N_{1}}\}$ and $\{\mathbf{y}\}=\{y_{1},y_{2},\ldots,y_{N_{2}}\}$.
\end{theorem}

\begin{proof}
For $N_{1}=N_{2}=1$, according to the Eq.(\ref{boson annihilation relations}), we have
\begin{eqnarray}
&&\mathbf{B}^{\ast}_{2}(y_{1})\mathbf{B}^{\ast}_{1}(x_{1})|\text{vac}\rangle\notag\\
&=&\rprod\limits
_{i\in[1,M_{2}]}\exp(y_{1}a^{\prime}_{-i}a^{\prime\ast}_{i-1})\rprod\limits
_{j\in[1,M_{1}]}\exp(x_{1}a_{-j}a^{\ast}_{j-1})|\text{vac}\rangle\notag\\
&=&\exp(y_{1}a^{\prime}_{-M_{2}}a^{\prime\ast}_{M_{2}-1})
\exp(y_{1}a^{\prime}_{-M_{2}+1}a^{\prime\ast}_{M_{2}-2})\rprod\limits
_{i\in[1,M_{2}-2]}\exp(y_{1}a^{\prime}_{-i}a^{\prime\ast}_{i-1})\exp(x_{1}a_{-M_{1}}
a^{\ast}_{M_{1}-1})\notag\\
&&\exp(x_{1}a_{-M_{1}+1}a^{\ast}_{M_{1}-2})\rprod\limits
_{j\in[1,M_{1}-2]}\exp(x_{1}a_{-j}a^{\ast}_{j-1})|\text{vac}\rangle\notag\\
&=&\exp\left(y_{1}a^{\prime}_{-M_{2}+1}a^{\prime\ast}_{M_{2}-2}-y^{2}_{1}
a^{\prime}_{-M_{2}}a^{\prime\ast}_{M_{2}-2}\right)
\exp(y_{1}a^{\prime}_{-M_{2}}a^{\prime\ast}_{M_{2}-1})\rprod\limits
_{i\in[1,M_{2}-2]}\exp(y_{1}a^{\prime}_{-i}a^{\prime\ast}_{i-1})\notag\\
&&\exp\left(x_{1}a_{-M_{1}+1}a^{\ast}_{M_{1}-2}-x^{2}_{1}a_{-M_{1}}a^{\ast}_{M_{1}-2}\right)
\exp(x_{1}a_{-M_{1}}a^{\ast}_{M_{1}-1})\rprod\limits
_{j\in[1,M_{1}-2]}\exp(x_{1}a_{-j}a^{\ast}_{j-1})|\text{vac}\rangle\notag\\
&=&\exp\left(y_{1}a^{\prime}_{-M_{2}+1}a^{\prime\ast}_{M_{2}-2}-y^{2}_{1}
a^{\prime}_{-M_{2}}a^{\prime\ast}_{M_{2}-2}\right)\exp(y_{1}a^{\prime}_{-M_{2}+2}
a^{\prime\ast}_{M_{2}-3})\rprod\limits
_{i\in[1,M_{2}-3]}\exp(y_{1}a^{\prime}_{-i}a^{\prime\ast}_{i-1})\notag\\
&&\exp\left(x_{1}a_{-M_{1}+1}a^{\ast}_{M_{1}-2}-x^{2}_{1}a_{-M_{1}}a^{\ast}_{M_{1}-2}\right)
\exp(x_{1}a_{-M_{1}+2}a^{\ast}_{M_{1}-3})\rprod\limits
_{j\in[1,M_{1}-3]}\exp(x_{1}a_{-j}a^{\ast}_{j-1})|\text{vac}\rangle\notag\\
&=&\exp\left(y_{1}a^{\prime}_{-M_{2}+2}a^{\prime\ast}_{M_{2}-3}-y^{2}_{1}
a^{\prime}_{-M_{2}+1}a^{\prime\ast}_{M_{2}-3}+y^{3}_{1}a^{\prime}_{-M_{2}}
a^{\prime\ast}_{M_{2}-3}\right)\rprod\limits
_{i\in[1,M_{2}-3]}\exp(y_{1}a^{\prime}_{-i}a^{\prime\ast}_{i-1})\notag\\
&&\exp\left(x_{1}a_{-M_{1}+2}a^{\ast}_{M_{1}-3}-x^{2}_{1}a_{-M_{1}+1}a^{\ast}_{M_{1}-3}
+x^{2}_{1}a_{-M_{1}}a^{\ast}_{M_{1}-3}\right)\rprod\limits
_{j\in[1,M_{1}-3]}\exp(x_{1}a_{-j}a^{\ast}_{j-1})|\text{vac}\rangle\notag\\
&=&\cdots\notag\\
&=&\exp\left(\sum\limits_{1\leq k\leq M_{2}}(-1)^{k-1}s_{(k)}\{y_{1}\}a^{\prime}_{-k}a^{\prime\ast}_{0}\right)
\exp\left(\sum\limits_{1\leq l\leq M_{1}}(-1)^{l-1}s_{(l)}\{x_{1}\}a_{-l}a^{\ast}_{0}\right)|\text{vac}\rangle.
\end{eqnarray}
Assume that the Eq.(\ref{B_1_N}) holds for $N_{1}-1$ and $N_{2}-1$, that is
\begin{eqnarray}
&&\mathbf{B}^{\ast}_{2}(y_{1})\mathbf{B}^{\ast}_{2}(y_{2})\cdots
\mathbf{B}^{\ast}_{2}(y_{N_{2}-1})\mathbf{B}^{\ast}_{1}(x_{1})
\mathbf{B}^{\ast}_{1}(x_{2})\cdots\mathbf{B}^{\ast}_{1}(x_{N_{1}-1})|\text{vac}\rangle\notag\\
&=&\exp\left(\sum\limits_{1\leq k\leq M_{2}\atop{1\leq n\leq N_{2}-1}}(-1)^{k-1}s_{(k,1^{n-1})}\{\mathbf{\overline{y}}\}a^{\prime}_{-k}
a^{\prime\ast}_{1-n}\right)
\exp\left(\sum\limits_{1\leq l\leq M_{1}\atop{1\leq n\leq N_{1}-1}}(-1)^{l-1}s_{(l,1^{n-1})}\{\mathbf{\overline{x}}\}a_{-l}a^{\ast}_{1-n}\right)|\text{vac}\rangle,
\notag\\
\end{eqnarray}
where $\{\mathbf{\overline{x}}\}=\{x_1,\ldots,x_{N_{1}-1}\}$ and $\{\mathbf{\overline{y}}\}=\{y_1,\ldots,y_{N_{2}-1}\}$.

Further, it follows that
\begin{eqnarray}\label{BBBB_N}
&&\mathbf{B}^{\ast}_{2}(y_{1})\mathbf{B}^{\ast}_{2}(y_{2})\cdots
\mathbf{B}^{\ast}_{2}(y_{N_{2}})\mathbf{B}^{\ast}_{1}(x_{1})
\mathbf{B}^{\ast}_{1}(x_{2})\cdots\mathbf{B}^{\ast}_{1}(x_{N_{1}})|\text{vac}\rangle\notag\\
&=&\mathbf{B}^{\ast}_{2}(y_{N_{2}})\mathbf{B}^{\ast}_{1}(x_{N_{1}})
\mathbf{B}^{\ast}_{2}(y_{1})\cdots
\mathbf{B}^{\ast}_{2}(y_{N_{2}-1})\mathbf{B}^{\ast}_{1}(x_{1})
\cdots\mathbf{B}^{\ast}_{1}(x_{N_{1}-1})|\text{vac}\rangle\notag\\
&=&\rprod\limits
_{i\in[2-N_{2},M_{2}]}\exp(y_{N_{2}}a^{\prime}_{-i}a^{\prime\ast}_{i-1})
\exp\left(\sum\limits_{1\leq k\leq M_{2}\atop{1\leq n\leq N_{2}-1}}(-1)^{k-1}s_{(k,1^{n-1})}\{\mathbf{\overline{y}}\}a^{\prime}_{-k}
a^{\prime\ast}_{1-n}\right)\notag\\
&&\rprod\limits_{j\in[2-N_{1},M_{1}]}\exp(x_{N_{1}}a_{-j}a^{\ast}_{j-1})
\exp\left(\sum\limits_{1\leq l\leq M_{1}\atop{1\leq n\leq N_{1}-1}}(-1)^{l-1}s_{(l,1^{n-1})}\{\mathbf{\overline{x}}\}a_{-l}a^{\ast}_{1-n}\right)
|\text{vac}\rangle\notag\\
&=&\rprod\limits
_{i\in[3-N_{2},M_{2}]}\exp(y_{N_{2}}a^{\prime}_{-i}a^{\prime\ast}_{i-1})
\exp\left(\sum\limits_{1\leq k\leq M_{2}\atop{1\leq n\leq N_{2}-2}}(-1)^{k-1}s_{(k,1^{n-1})}\{\mathbf{\overline{y}}\}a^{\prime}_{-k}
a^{\prime\ast}_{1-n}\right)\notag\\
&&\exp\left(y_{N_{2}}a^{\prime}_{N_{2}-2}a^{\prime\ast}_{1-N_{2}}\right)
\exp\left(\sum\limits_{1\leq k\leq M_{2}}
(-1)^{k-1}s_{(k,1^{N_{2}-2})}\{\mathbf{\overline{y}}\}a^{\prime}_{-k}
a^{\prime\ast}_{2-N_{2}}\right)\notag\\
&&\rprod\limits_{j\in[3-N_{1},M_{1}]}
\exp(x_{N_{1}}a_{-j}a^{\ast}_{j-1})\exp\left(\sum\limits_{1\leq l\leq M_{1}\atop{1\leq n\leq N_{1}-2}}(-1)^{l-1}s_{(l,1^{n-1})}\{\mathbf{\overline{x}}\}a_{-l}a^{\ast}_{1-n}\right)\notag\\
&&\exp\left(x_{N_{1}}a_{N_{1}-2}a^{\ast}_{1-N_{1}}\right)\exp\left(\sum\limits_{1\leq l\leq M_{1}}(-1)^{l-1}s_{(l,1^{N_{1}-2})}\{\mathbf{\overline{x}}\}a_{-l}a^{\ast}_{2-N_{1}}\right)
|\text{vac}\rangle\notag\\
&=&\rprod\limits
_{i\in[3-N_{2},M_{2}]}\exp(y_{N_{2}}a^{\prime}_{-i}a^{\prime\ast}_{i-1})
\exp\left(\sum\limits_{1\leq k\leq M_{2}\atop{1\leq n\leq N_{2}-2}}(-1)^{k-1}s_{(k,1^{n-1})}\{\mathbf{\overline{y}}\}a^{\prime}_{-k}
a^{\prime\ast}_{1-n}\right)\notag\\
&&\exp\left(\sum\limits_{1\leq k\leq M_{2}}(-1)^{k-1}s_{(k,1^{N_{2}-2})}\{\mathbf{\overline{y}}\}a^{\prime}_{-k}\left(
a^{\prime\ast}_{2-N_{2}}+y_{N_{2}}a^{\prime\ast}_{1-N_{2}}\right)\right)
\exp\left(y_{N_{2}}a^{\prime}_{N_{2}-2}a^{\prime\ast}_{1-N_{2}}\right)\notag\\
&&\rprod\limits_{j\in[3-N_{1},M_{1}]}
\exp(x_{N_{1}}a_{-j}a^{\ast}_{j-1})\exp\left(\sum\limits_{1\leq l\leq M_{1}\atop{1\leq n\leq N_{1}-2}}(-1)^{l-1}s_{(l,1^{n-1})}\{\mathbf{\overline{x}}\}a_{-l}a^{\ast}_{1-n}\right)\notag\\
&&\exp\left(\sum\limits_{1\leq l\leq M_{1}}(-1)^{l-1}s_{(l,1^{N_{1}-2})}
\{\mathbf{\overline{x}}\}a_{-l}\left(a^{\ast}_{2-N_{1}}+x_{N_{1}}a^{\ast}_{1-N_{1}}\right)\right)
\exp\left(x_{N_{1}}a_{N_{1}-2}a^{\ast}_{1-N_{1}}\right)|\text{vac}\rangle\notag\\
&=&\rprod\limits
_{i\in[1,M_{2}]}\exp(y_{N_{2}}a^{\prime}_{-i}a^{\prime\ast}_{i-1})
\exp\left(\sum\limits_{1\leq k\leq M_{2}\atop{1\leq n\leq N_{2}-1}}(-1)^{k-1}s_{(k,1^{n-1})}\{\mathbf{\overline{y}}\}a^{\prime}_{-k}\left(
a^{\prime\ast}_{1-n}+y_{N_{2}}a^{\prime\ast}_{-n}\right)\right)\notag\\
&&\rprod\limits_{j\in[1,M_{1}]}\exp(x_{N_{1}}a_{-j}a^{\ast}_{j-1})
\exp\left(\sum\limits_{1\leq l\leq M_{1}\atop{1\leq n\leq N_{1}-1}}(-1)^{l-1}s_{(l,1^{n-1})}\{\mathbf{\overline{x}}\}a_{-l}\left(a^{\ast}_{1-n}
+x_{N_{1}}a^{\ast}_{-n}\right)\right)|\text{vac}\rangle.\notag\\
\end{eqnarray}
For simplicity, we introduce new notations $A^{\prime}=\sum\limits_{1\leq k\leq M_{2}}A^{\prime}_{k}$ and $A=\sum\limits_{1\leq l\leq M_{1}}A_{l}$, where
\begin{eqnarray}
&&A^{\prime}_{k}=\sum\limits_{1\leq n\leq N_{2}-1}(-1)^{k-1}s_{(k,1^{n-1})}\{\mathbf{\overline{y}}\}a^{\prime}_{-k}\left(
a^{\prime\ast}_{1-n}+y_{N_{2}}a^{\prime\ast}_{-n}\right),\notag\\
&&A_{l}=\sum\limits_{1\leq n\leq N_{1}-1}(-1)^{l-1}s_{(l,1^{n-1})}\{\mathbf{\overline{x}}\}a_{-l}\left(a^{\ast}_{1-n}
+x_{N_{1}}a^{\ast}_{-n}\right).
\end{eqnarray}
Denote $B^{\prime}_{1}=y_{N_{2}}a^{\prime}_{-1}a^{\prime\ast}_{0}+A^{\prime}_{1}$, $B_{1}=x_{N_{1}}a_{-1}a^{\ast}_{0}+A_{1}$,
$B^{\prime}_{p}=\left[y_{N_{2}}a^{\prime}_{-p}a^{\prime\ast}_{p-1},
B^{\prime}_{p-1}\right]+A^{\prime}_{p}\ (2\leq p\leq M_{2})$ and $B_{q}=[x_{N_{1}}a_{-q}a^{\ast}_{q-1},B_{q-1}]+A_{q-1}\ (2\leq q\leq M_{1})$. Owing to the Eqs.(\ref{schur-skew-schur}) and (\ref{skew-schur}), it can be checked that
\begin{eqnarray}
B^{\prime}_{1}&=&y_{N_{2}}a^{\prime}_{-1}a^{\prime\ast}_{0}+\sum\limits_{1\leq n\leq N_{2}-1}s_{(1,1^{n-1})}\{\mathbf{\overline{y}}\}a^{\prime}_{-1}\left(
a^{\prime\ast}_{1-n}+y_{N_{2}}a^{\prime\ast}_{-n}\right)\notag\\
&=&\left(y_{N_{2}}+s_{(1)}\{\mathbf{\overline{y}}\}\right)a^{\prime}_{-1}a^{\prime\ast}_{0}
+\sum\limits_{2\leq n\leq N_{2}-1}\left(s_{(1^{n})}\{\mathbf{\overline{y}}\}+s_{(1^{n-1})}
\{\mathbf{\overline{y}}\}y_{N_{2}}\right)a^{\prime}_{-1}a^{\prime\ast}_{1-n}\notag\\
&&+
s_{(1^{N_{2}-1})}\{\mathbf{\overline{y}}\}y_{N_{2}}a^{\prime}_{-1}a^{\prime\ast}_{1-N_{2}}\notag\\
&=&\sum\limits_{1\leq n\leq N_{2}}\left(s_{(1^{n})}\{\mathbf{\overline{y}}\}+s_{(1^{n-1})}
\{\mathbf{\overline{y}}\}y_{N_{2}}\right)a^{\prime}_{-1}a^{\prime\ast}_{1-n}\notag\\
&=&\sum\limits_{1\leq n\leq N_{2}}s_{(1,1^{n-1})}\{\mathbf{y}\}a^{\prime}_{-1}a^{\prime\ast}_{1-n}.
\end{eqnarray}
and
\begin{eqnarray}
B^{\prime}_{2}&=&\big[y_{N_{2}}a^{\prime}_{-2}a^{\prime\ast}_{1},\sum\limits_{1\leq n\leq N_{2}}\left(s_{(1^{n})}\{\mathbf{\overline{y}}\}+s_{(1^{n-1})}
\{\mathbf{\overline{y}}\}y_{N_{2}}\right)a^{\prime}_{-1}a^{\prime\ast}_{1-n}\big]\notag\\
&&-\sum\limits_{1\leq n\leq N_{2}-1}s_{(2,1^{n-1})}\{\mathbf{\overline{y}}\}a^{\prime}_{-2}\left(
a^{\prime\ast}_{1-n}+y_{N_{2}}a^{\prime\ast}_{-n}\right)\notag\\
&=&-\sum\limits_{1\leq n\leq N_{2}}\left(s_{(1^{n})}\{\mathbf{\overline{y}}\}y_{N_{2}}+
s_{(1^{n-1})}\{\mathbf{\overline{y}}\}y^{2}_{N_{2}}\right)a^{\prime}_{-2}a^{\prime\ast}_{1-n}
-\sum\limits_{1\leq n\leq N_{2}-1}s_{(2,1^{n-1})}\{\mathbf{\overline{y}}\}a^{\prime}_{-2}a^{\prime\ast}_{1-n}\notag\\
&&-\sum\limits_{1\leq n\leq N_{2}-1}s_{(2,1^{n-1})}\{\mathbf{\overline{y}}\}y_{N_{2}}
a^{\prime}_{-2}a^{\prime\ast}_{-n}\notag\\
&=&-\sum\limits_{1\leq n\leq N_{2}}\left(s_{(1^{n})}\{\mathbf{\overline{y}}\}y_{N_{2}}+
s_{(1^{n-1})}\{\mathbf{\overline{y}}\}y^{2}_{N_{2}}\right)a^{\prime}_{-2}a^{\prime\ast}_{1-n}
-\sum\limits_{1\leq n\leq N_{2}-1}s_{(2,1^{n-1})}\{\mathbf{\overline{y}}\}a^{\prime}_{-2}a^{\prime\ast}_{1-n}\notag\\
&&-\sum\limits_{2\leq n\leq N_{2}}s_{(2,1^{n-2})}\{\mathbf{\overline{y}}\}y_{N_{2}}
a^{\prime}_{-2}a^{\prime\ast}_{1-n}\notag\\
&=&-\left(s_{(0)}\{\mathbf{\overline{y}}\}y^{2}_{N_{2}}+
s_{(1)}\{\mathbf{\overline{y}}\}y_{N_{2}}+s_{(2)}\{\mathbf{\overline{y}}\}\right)
a^{\prime}_{-2}a^{\prime\ast}_{0}-\sum\limits_{2\leq n\leq N_{2}-1}\big(s_{(1^{n-1})}\{\mathbf{\overline{y}}\}y^{2}_{N_{2}}
+s_{(2,1^{n-2})}\{\mathbf{\overline{y}}\}
y_{N_{2}}\notag\\
&&+s_{(1^{n})}\{\mathbf{\overline{y}}\}y_{N_{2}}+s_{(2,1^{n-1})}
\{\mathbf{\overline{y}}\}\big)a^{\prime}_{-2}a^{\prime\ast}_{1-n}-
\big(s_{(2,1^{n-2})}\{\mathbf{\overline{y}}\}
y_{N_{2}}+s_{(1^{N_{2}-1})}\{\mathbf{\overline{y}}\}y^{2}_{N_{2}}\notag\\
&&+s_{(1^{N_{2}})}\{\mathbf{\overline{y}}\}y_{N_{2}}\big)
a^{\prime}_{-2}a^{\prime\ast}_{1-N_{2}}\notag\\
&=&-\sum\limits_{1\leq n\leq N_{2}}s_{(2,1^{n-1})}\{\mathbf{y}\}a^{\prime}_{-2}a^{\prime\ast}_{1-n}.
\end{eqnarray}
By the same method, one obtains
\begin{eqnarray}
&&B_{1}=\sum\limits_{1\leq n\leq N_{1}}s_{(1,1^{n-1})}\{\mathbf{x}\}a_{-1}a^{\ast}_{1-n},\notag\\
&&B_{2}=-\sum\limits_{1\leq n\leq N_{2}}s_{(2,1^{n-1})}\{\mathbf{x}\}a_{-2}a^{\ast}_{1-n}.
\end{eqnarray}
Proceeding similarly, it is straightforward to show that
\begin{eqnarray}
&&B^{\prime}_{p}=(-1)^{p-1}\sum\limits_{1\leq n\leq N_{2}}s_{(p,1^{n-1})}\{\mathbf{y}\}a^{\prime}_{-p}a^{\prime\ast}_{1-n},\quad 1\leq p\leq M_{2},\notag\\
&&B_{q}=(-1)^{q-1}\sum\limits_{1\leq n\leq N_{1}}s_{(q,1^{n-1})}\{\mathbf{x}\}a_{-q}a^{\ast}_{1-n},\quad 1\leq q\leq M_{1}.
\end{eqnarray}
Thus, the Eq. (\ref{BBBB_N}) simplifies to
\begin{eqnarray}
&&\rprod\limits_{i\in[1,M_{2}]}\exp(y_{N_{2}}a^{\prime}_{-i}a^{\prime\ast}_{i-1})
\exp\left(A^{\prime}\right)
\rprod\limits_{j\in[1,M_{1}]}\exp(x_{N_{1}}a_{-j}a^{\ast}_{j-1})\exp\left(A\right)
|\text{vac}\rangle\notag\\
&=&\rprod\limits_{i\in[2,M_{2}]}\exp(y_{N_{2}}a^{\prime}_{-i}a^{\prime\ast}_{i-1})
\exp\left(y_{N_{2}}a^{\prime}_{-1}a^{\prime\ast}_{0}\right)\exp\left(A^{\prime}_{1}
+A^{\prime}_{2}+\cdots+A^{\prime}_{M_{2}}\right)
\rprod\limits_{j\in[2,M_{1}]}\exp(x_{N_{1}}a_{-j}a^{\ast}_{j-1})\notag\\
&&\exp\left(x_{N_{1}}a_{-1}a^{\ast}_{0}\right)\exp\left(A_{1}
+A_{2}+\cdots+A_{M_{1}}\right)|\text{vac}\rangle\notag\\
&=&\rprod\limits_{i\in[2,M_{2}]}\exp(y_{N_{2}}a^{\prime}_{-i}a^{\prime\ast}_{i-1})
\exp\left(B^{\prime}_{1}\right)\exp\left(A^{\prime}_{2}+\cdots+A^{\prime}_{M_{2}}\right)
\rprod\limits_{j\in[2,M_{1}]}\exp(x_{N_{1}}a_{-j}a^{\ast}_{j-1})
\exp\left(B_{1}\right)\notag\\
&&\exp\left(A_{2}+\cdots+A_{M_{1}}\right)|\text{vac}\rangle\notag\\
&=&\rprod\limits_{i\in[3,M_{2}]}\exp(y_{N_{2}}a^{\prime}_{-i}a^{\prime\ast}_{i-1})
\exp\left(B^{\prime}_{1}+B^{\prime}_{2}+A^{\prime}_{3}+\cdots+A^{\prime}_{M_{2}}\right)
\rprod\limits_{j\in[3,M_{1}]}\exp(x_{N_{1}}a_{-j}a^{\ast}_{j-1})\notag\\
&&\exp\left(B_{1}+B_{2}+A_{3}+\cdots+A_{M_{1}}\right)|\text{vac}\rangle\notag\\
&=&\cdots\notag\\
&=&\exp\left(\sum\limits_{1\leq k\leq M_{2}}B^{\prime}_{k}\right)
\exp\left(\sum\limits_{1\leq l\leq M_{1}}B_{l}\right)|\text{vac}\rangle\notag\\
&=&\exp\left(\sum\limits_{1\leq k\leq M_{2}\atop{1\leq n\leq N_{2}}}(-1)^{k-1}s_{(k,1^{n-1})}\{\mathbf{y}\}a^{\prime}_{-k}a^{\prime\ast}_{1-n}\right)
\exp\left(\sum\limits_{1\leq l\leq M_{1}\atop{1\leq n\leq N_{1}}}(-1)^{l-1}s_{(l,1^{n-1})}\{\mathbf{x}\}a_{-l}a^{\ast}_{1-n}\right)|\text{vac}\rangle.
\notag\\
\end{eqnarray}
This completes the proof.
\end{proof}

\begin{corollary}
For $\forall M_{1}, M_{2}\in\mathbb{Z}_{+}$, it also can be derived that
\begin{align}
&\langle\text{vac}|\mathbf{C}^{\ast}_{1}(x_{N_{1}})\mathbf{C}^{\ast}_{1}(x_{N_{1}-1})
\cdots\cdots\mathbf{C}^{\ast}_{1}(x_{1})\mathbf{C}^{\ast}_{2}(y_{N_{2}})
\mathbf{C}^{\ast}_{2}(y_{N_{2}-1})\cdots\mathbf{C}^{\ast}_{2}(y_{1})\notag\\
&=\langle\text{vac}|\exp\left(\sum\limits_{1\leq l\leq M_{1}\atop{1\leq n\leq N_{1}}}(-1)^{l-1}s_{(l,1^{n-1})}\{\mathbf{x}\}a_{n-1}a^{\ast}_{l}\right)
\exp\left(\sum\limits_{1\leq k\leq M_{2}\atop{1\leq n\leq N_{2}}}(-1)^{k-1}s_{(k,1^{n-1})}\{\mathbf{y}\}a^{\prime}_{n-1}a^{\prime\ast}_{k}\right).
\end{align}
\end{corollary}

\begin{theorem}
For any positive integers $M_{1},M_{2},N_{1}$ and $N_{2}$, the correlation function of the bosonic UC hierarchy is
\begin{eqnarray}
&&\langle\text{vac}|\mathbf{C}^{\ast}_{1}(x_{N_{1}})
\cdots\cdots\mathbf{C}^{\ast}_{1}(x_{1})\mathbf{C}^{\ast}_{2}(y_{N_{2}})
\cdots\mathbf{C}^{\ast}_{2}(y_{1})\mathbf{B}^{\ast}_{2}(v_{1})\cdots
\mathbf{B}^{\ast}_{2}(v_{N_{2}})\mathbf{B}^{\ast}_{1}(u_{1})
\cdots\mathbf{B}^{\ast}_{1}(u_{N_{1}})|\text{vac}\rangle\notag\\
&=&\frac{1}{\sum\limits_{\mu\subseteq[N_{1},M_{1}]\atop{\lambda\subseteq[N_{2},M_{2}]}}
\sum\limits_{\tau,\nu,\xi\in\mathcal{P}\atop{\widetilde{\tau},\widetilde{\nu},
\widetilde{\xi}\in\mathcal{P}}}
C^{\mu}_{\nu,\tau}C^{\lambda}_{\xi,\tau}C^{\mu}_{\widetilde{\nu},\widetilde{\tau}}
C^{\lambda}_{\widetilde{\xi},\widetilde{\tau}}
S_{[\nu,\xi]}\{\mathbf{x},\mathbf{y}\}
S_{[\widetilde{\nu},\widetilde{\xi}]}\{\mathbf{u},\mathbf{v}\}}.
\end{eqnarray}
\end{theorem}

\begin{proof}
Before proving the main theorem,  it is necessary to introduce the following auxiliary lemma.
\begin{lemma}\textsuperscript{\cite{bosoncorreclation}}
Let $M=(M_{ij})_{1\leq i,j\leq n}$ be an $n\times n$ matrix, and define $W_{k}=\sum\limits^{n}_{i=1}M_{ik}\frac{w_{k}}{w_{i}}$ for each $k\in\{1,2,\ldots,n\}$. Then the constant term $CT(F)$ of  the rational function $F=\prod\limits^{n}_{k=1}\frac{1}{W_{k}}$  with respect to the variables $w_{i}$ equals $\frac{1}{\det M}$.
\end{lemma}
Based on the above lemma, we now proceed to prove the main theorem. For any positive integers $M_{1}, M_{2},N_{1}, N_{2}$, denote
\begin{eqnarray}
&&X^{\prime}_{k}=\sum\limits_{1\leq n\leq N_{2}}s_{(k,1^{n-1})}\{\mathbf{y}\}a^{\prime}_{n-1},\quad Y^{\prime}_{k}=\sum\limits_{1\leq n\leq N_{2}}s_{(k,1^{n-1})}\{\mathbf{v}\}a^{\prime\ast}_{1-n}
,\quad 1\leq k\leq M_{2}\notag\\
&&X_{l}=\sum\limits_{1\leq n\leq N_{1}}s_{(l,1^{n-1})}\{\mathbf{x}\}a_{n-1},\quad \ Y_{l}=\sum\limits_{1\leq n\leq N_{1}}s_{(l,1^{n-1})}\{\mathbf{u}\}a^{\ast}_{1-n},\quad \ 1\leq l\leq M_{1}.
\end{eqnarray}
It follows from Eq.(\ref{boson_commutation}) that
\begin{eqnarray}
&&M_{ij}=[X_{i},Y_{j}]=\sum\limits_{1\leq n\leq N_{1}}s_{(i,1^{n-1})}\{\mathbf{x}\}s_{(j,1^{n-1})}\{\mathbf{u}\},\notag\\
&&M^{\prime}_{ij}=[X^{\prime}_{i},Y^{\prime}_{j}]=\sum\limits_{1\leq n\leq N_{2}}s_{(i,1^{n-1})}\{\mathbf{y}\}s_{(j,1^{n-1})}\{\mathbf{v}\}.
\end{eqnarray}
Associated with the matrices $M=(M_{ij})_{1\leq i,j\leq M_{1}}$ and $M^{\prime}=(M^{\prime}_{ij})_{1\leq i,j\leq M_{2}}$, introduce the functions
$W_{j}=\sum\limits^{M_{1}}_{i=1}\left(\delta_{ij}+M_{ij}\right)\frac{w_{j}}{w_{i}}(1\leq j\leq M_{1})$ and $W^{\prime}_{j^{\prime}}=\sum\limits^{M_{2}}_{i^{\prime}=1}
\left(\delta_{i^{\prime}j^{\prime}}+M^{\prime}_{i^{\prime}j^{\prime}}\right)
\frac{w^{\prime}_{j^{\prime}}}{w^{\prime}_{i^{\prime}}}(1\leq j^{\prime}\leq M_{2})$. It can be verified that the identity
\begin{eqnarray}
\det(I+M)=\sum\limits_{\nu\subseteq[N_{1},M_{1}]}s_{\nu}\{\mathbf{x}\}
s_{\nu}\{\mathbf{u}\}
\end{eqnarray}
holds.

By the Eqs.(\ref{boson_commutation}) and (\ref{double_product}), it is seen that
\begin{eqnarray}
&&\langle\text{vac}|\mathbf{C}^{\ast}_{1}(x_{N_{1}})
\cdots\cdots\mathbf{C}^{\ast}_{1}(x_{1})\mathbf{C}^{\ast}_{2}(y_{N_{2}})
\cdots\mathbf{C}^{\ast}_{2}(y_{1})\mathbf{B}^{\ast}_{2}(v_{1})\cdots
\mathbf{B}^{\ast}_{2}(v_{N_{2}})\mathbf{B}^{\ast}_{1}(u_{1})
\cdots\mathbf{B}^{\ast}_{1}(u_{N_{1}})|\text{vac}\rangle\notag\\
&=&\langle\text{vac}|\exp\left(\sum\limits_{1\leq l\leq M_{1}}(-1)^{l-1}X_{l}a^{\ast}_{l}\right)\exp\left(\sum\limits_{1\leq k\leq M_{2}}(-1)^{k-1}X^{\prime}_{k}a^{\prime\ast}_{k}\right)
\exp\left(\sum\limits_{1\leq k\leq M_{2}}(-1)^{k-1}Y^{\prime}_{k}a^{\prime}_{-k}\right)\notag\\
&&\exp\left(\sum\limits_{1\leq l\leq M_{1}}(-1)^{l-1}Y_{l}a_{-l}\right)
|\text{vac}\rangle\notag\\
&=&\langle\text{vac}|\sum\limits_{k_{l}\geq0}\prod\limits^{M_{1}}_{l=1}
\frac{\left((-1)^{l-1}X_{l}a^{\ast}_{l}\right)^{k_{l}}}{k_{l}!}
\sum\limits_{k_{p}\geq0}\prod\limits^{M_{2}}_{p=1}
\frac{\left((-1)^{p-1}X^{\prime}_{p}a^{\prime\ast}_{p}\right)^{k_{p}}}{k_{p}!}
\sum\limits_{m_{p}\geq0}\prod\limits^{M_{2}}_{p=1}
\frac{\left((-1)^{p-1}Y^{\prime}_{p}a^{\prime}_{-p}\right)^{m_{p}}}{m_{p}!}\notag\\
&&\sum\limits_{m_{l}\geq0}\prod\limits^{M_{1}}_{l=1}
\frac{\left((-1)^{l-1}Y_{l}a_{-l}\right)^{m_{l}}}{m_{l}!}|\text{vac}\rangle\notag\\
&=&\sum\limits_{k_{l}\geq0}\sum\limits_{k_{p}\geq0}\langle\text{vac}|
\prod\limits^{M_{1}}_{l=1}\frac{\left(X_{l}\right)^{k_{l}}}{k_{l}!}
\prod\limits^{M_{2}}_{p=1}\frac{\left(X^{\prime}_{p}\right)^{k_{p}}}{k_{p}!}
\prod\limits^{M_{2}}_{p=1}\frac{\left(Y^{\prime}_{p}\right)^{k_{p}}}{k_{p}!}
\prod\limits^{M_{1}}_{l=1}\frac{\left(Y_{l}\right)^{k_{l}}}{k_{l}!}|\text{vac}\rangle
\cdot(-1)^{\sum\limits^{M_{1}}_{l=1}k_{l}+\sum\limits^{M_{2}}_{p=1}k_{p}}
\prod\limits^{M_{1}}_{l=1}k_{l}!\prod\limits^{M_{2}}_{p=1}k_{p}!\notag\\
&=&\sum\limits_{k_{l}\geq0}\sum\limits_{k_{p}\geq0}\langle\text{vac}|
\prod\limits^{M_{1}}_{l=1}\frac{\left(-X_{l}\right)^{k_{l}}}{k_{l}!}
\prod\limits^{M_{2}}_{p=1}\frac{\left(-X^{\prime}_{p}\right)^{k_{p}}}{k_{p}!}
\prod\limits^{M_{2}}_{p=1}\left(Y^{\prime}_{p}\right)^{k_{p}}
\prod\limits^{M_{1}}_{l=1}\left(Y_{l}\right)^{k_{l}}|\text{vac}\rangle\notag\\
&=&CT\left(\langle\text{vac}|\prod\limits^{M_{1}}_{i=1}\exp\left(-\frac{X_{i}}{w_{i}}\right)
\prod\limits^{M_{1}}_{j=1}\frac{1}{1-w_{j}Y_{j}}\prod\limits^{M_{2}}_{i^{\prime}=1}
\exp\left(-\frac{X^{\prime}_{i^{\prime}}}{w^{\prime}_{i^{\prime}}}\right)
\prod\limits^{M_{2}}_{j^{\prime}=1}\frac{1}{1-w^{\prime}_{j^{\prime}}Y^{\prime}_{j^{\prime}}}
\right)\notag\\
&=&CT\left(\prod\limits^{M_{1}}_{j=1}\frac{1}{W_{j}}\right)
CT\left(\prod\limits^{M_{2}}_{j^{\prime}=1}\frac{1}{W^{\prime}_{j^{\prime}}}\right)\notag\\
&=&\frac{1}{\det(I+M)\det(I+M^{\prime})}\notag\\
&=&\frac{1}{\sum\limits_{\mu\subseteq[N_{1},M_{1}]\atop{\lambda\subseteq[N_{2},M_{2}]}}
s_{\mu}\{\mathbf{x}\}s_{\lambda}\{\mathbf{y}\}s_{\mu}\{\mathbf{u}\}
s_{\lambda}\{\mathbf{v}\}}\notag\\
&=&\frac{1}{\sum\limits_{\mu\subseteq[N_{1},M_{1}]\atop{\lambda\subseteq[N_{2},M_{2}]}}
\sum\limits_{\tau,\nu,\xi\in\mathcal{P}\atop{\widetilde{\tau},\widetilde{\nu},
\widetilde{\xi}\in\mathcal{P}}}
C^{\mu}_{\nu,\tau}C^{\lambda}_{\xi,\tau}C^{\mu}_{\widetilde{\nu},\widetilde{\tau}}
C^{\lambda}_{\widetilde{\xi},\widetilde{\tau}}
S_{[\nu,\xi]}\{\mathbf{x},\mathbf{y}\}
S_{[\widetilde{\nu},\widetilde{\xi}]}\{\mathbf{u},\mathbf{v}\}}.
\end{eqnarray}
This completes the proof.
\end{proof}

\section{Summary and discussion}
We have constructed the bosonic UC hierarchy within the framework of charged free bosons, and investigated its tau functions and correlation functions. It is shown that the tau functions of the bosonic UC hierarchy can be explicitly expressed by a series of ordered exponential operators acting on the vacuum vector. Moreover, the correlation functions of the bosonic UC hierarchy  are given by a sum of products of UCs. The representation of the Lie algebra $\mathfrak{\widehat{gl}}_{2\infty}$ has also been developed. A notable observation is that the correlation functions of the bosonic UC hierarchy are precisely the inverse of those of the generalized phase model \cite{correclation_GPM}. As is well known, the $q$-boson model is closely related to Hall-Littlewood functions. It is an interesting question how to link the framework of charged free bosons with the q-boson model and Hall-Littlewood functions. It will be considered in the near future.

\section{Acknowledgements}
This article is dedicated to Professor Ke Wu in Capital Normal University in celebration of his 80th birthday. And this work is supported by the National Natural Science Foundation of China (Grant No. 12461048), the Natural Science Foundation of Inner Mongolia Autonomous Region (Grant No. 2023MS01003)  and the Inner Mongolia Autonomous Region Science and Technology Program Project (Grant No. 2025KYPT0098)


\begin{thebibliography}{99}

\bibitem{quantumphysics}R. J. Baxter, Exactly Solved Models in Statistical Mechanics, Academic Press, New York, 1982.


\bibitem{integrable_quantum}R. Donagi, B. Dubrovin, E. Frenkel and E. Previato, Integrable Systems and Quantum Groups, Springer Berlin, Heidelberg, 1993.


\bibitem{introduction_integrablesystem}O. Babelon, D. Bernard and M. Talon, Introduction to Classical Integrable Systems, Cambridge University Press, Cambridge, 2009.

\bibitem{integrable_wave}R. Camassa and D. D. Holm, An integrable shallow water equation with peaked solitons, Phys. Rev. Lett. 71 (1993) 1661-1664.


\bibitem{integrable_QFT}V. V. Bazhanov, S. L. Lukyanov and A. B. Zamolodchikov, Integrable structure of conformal field theory, quantum KdV theory and Thermodynamic Bethe Ansatz, Commun. Math. Phys. 177 (1996) 381-398.

\bibitem{integrable_string}R. Dijkgraaf and C. Vafa, Matrix models, topological strings, and supersymmetric gauge theories, Nucl. Phys. B 644 (2002) 3-20.


\bibitem{Jimbo82}E. Date, M. Jimbo, M. Kashiwara and T. Miwa, Transformation groups for soliton equations-Euclidean Lie algebras and reduction of the KP hierarchy, Publ. Res. Inst. Math. Sci. 18 (1982) 1077-1110.

\bibitem{Jimbo2}M. Jimbo and T. Miwa, Solitons and infinite-dimensional Lie algebras, Publ. Res. Inst. Math. Sci. 19 (1983) 943-1001.


\bibitem{Kac2013}V. G. Kac, A. K. Raina and N. Rozhkovskaya, Bombay Lectures on Highest Height Representations of Iinfinite Dimensional Lie Algebras, World Scientific Publishing, Singapore, 2013.

\bibitem{UC}T. Tsuda, Universal characters and an extension of the KP hierarchy, Commun. Math. Phys. 248 (2004) 501-526.



\bibitem{Wheeler}M. Wheeler, Free fermions in classical and quantum integrable models, Ph.D thesis, Univ. Melbourne, arXiv:1110.6703v1.


\bibitem{GPM2019}N. Wang and C. Z. Li, Universal character, phase model and topological strings on $\mathbb{C}^{3}$, Eur. Phys. J. C 79 (2019) 953.


\bibitem{Jimbo2000}M. Jimbo, T. Miwa and E. Date, Solitons: Differential Equations, Symmetries and Infinite Dimensional Algebras, Cambridge University Press, Cambridge, 2000.

\bibitem{boson_integrable}B. Bakalov and D. Fleisher, Bosonizations of $\mathfrak{\widehat{sl}}_{2}$ and integrable hierarchies, SIGMA 11 (2015) 005.


\bibitem{Walgrbra_boson}W. Q. Wang, $\mathcal{W}_{1+\infty}$ algebra, $\mathcal{W}_{3}$ algebra, and Friedan-Martinec-Shenker bosonization, Commun. Math. Phys. 195 (1998) 95-111.

\bibitem{boson_hierarchy}K. T. Liszewski, The charged free boson integrable hierarchy, Ph.D thesis, North Carolina State University, 2011.

\bibitem{tau_function_boson}N. H. Jing and Z. J. Li, Tau functions of the charged free bosons, Sci. China Math. 63 (2020) 2157-2176.

\bibitem{mBKP_fermoinic}W. C. Guo, M. Y. Chen, Y. Yang and J. P. Cheng, Darboux transformations of the modified BKP hierarchy by fermionic approach, J. Math. Phys. 64 (2023) 103501.

\bibitem{UC_fermion}Y. N. Wang and Z. W. Yan, Solutions of the universal character hierarchy and BUC hierarchy by fermionic approach, J. Math. Anal. Appl. 532 (2024) 127912.

\bibitem{PM_fermion}Z. N. Cui, Y. Bai, N. Wang and K. Wu, The fermion representation of the phase model, Chin. Q. J. Math. 37 (2022) 317.

\bibitem{GPM_fermion}X. Zhang and Z. W. Yan, The fermion representation of the generalized phase model, Nucl. Phys. B 1002 (2024) 116532.

\bibitem{Mac}I. G. Macdonald, Symmetric Functions and Hall Polynomials, Clarendon Press, Oxford, 1995.

\bibitem{Macdonald_integrable}H. Awata, S. Odake and J. Shiraishi, Integral representations of the Macdonald symmetric polynomials, Commun. Math. Phys. 179 (1996) 647-666.

\bibitem{Toda_KP}D. Yang and J. Zhou, From Toda hierarchy to KP hierarchy, SIGMA 21 (2025) 068.

\bibitem{symmetric_integrable}A. Mironov, A. Morozov and A. Popolitov, Symmetric polynomials: DIM integrable systems versus twisted Cherednik systems, Phys. Lett. B 877 (2026) 140457.


\bibitem{universal_characters}K. Koike, On the decomposition of tensor products of the representations of the classical groups: by means of the universal characters, Adv. Math. 74 (1989) 57-86.

\bibitem{correctionphase}N. M. Bogoliubov, A. G. Izergin and N. A. Kitanine, Correlation functions for a strongly correlated boson system, Nucl. Phys. B 516 (1998) 501-528.


\bibitem{model-symmetricfunction}N. V. Tsilevich, Quantum inverse scattering method for the $q$-boson model and symmetric functions, Funct. Anal. Appl. 40 (2006) 207-217.

\bibitem{correclation_GPM}D. H. Li, S. Y. Zhang and Z. W. Yan, Correlation functions of the generalized phase model, Phys. Lett. B 878 (2026) 140526.


\bibitem{bosoncorreclation}N. H. Jing, Z. J. Li and T. W. Cai, Correlation functions of charged free boson and fermion systems, J. Stat. Mech. 2020 (2020) 083101.


\bibitem{MKP_Boson}Y. Y. Zhang, J. P. Cheng, S. F. Shen and J. Hu, Modified bosonic integrable hierarchy, J. Geom. Phys. 201 (2024) 105199.


\bibitem{BCH_formula}R. C. Thompson, Cyclic relations and the Goldberg coefficients in the Campbell-Baker-Hausdorff formula, Proc. Amer. Math. Soc. 86 (1982) 12-14.


\end{thebibliography}
\end{document}